\documentclass[12pt]{article} 
\usepackage{amsmath,amsthm}
\usepackage{latexsym,amssymb}
\usepackage{graphicx, epsfig}
\usepackage[refpage]{nomencl} \makenomenclature
\usepackage{url}
\newcommand{\zz}{\mathbb{Z}}
 
\newcommand{\nn}{\mathbb{N}}
\newcommand{\rr}{\mathbb{R}}

\newcommand{\mca}{\mathcal{A}}

\newcommand{\mcc}{\mathcal{C}}

\newcommand{\mcg}{\mathcal{G}}

\newcommand{\mct}{\mathcal{T}}
\newcommand{\mcu}{\mathcal{U}}

\newcommand{\pr}{\mathtt{Pr}}
\newcommand{\cond}{\mid} 
\numberwithin{equation}{section}

\theoremstyle{plain}
\newtheorem{thm}{Theorem}[section]
\newtheorem{cor}[thm]{Corollary}
\newtheorem{lemma}[thm]{Lemma}
\newtheorem{prop}[thm]{Proposition}

\theoremstyle{definition}
\newtheorem{defn}[thm]{Definition}

\newtheorem{problem}[thm]{Problem}

\theoremstyle{remark}
\newtheorem{rem}[thm]{Remark}

\title{Reconstructing pedigrees: some identifiability questions for a
  recombination-mutation model}
\author{Bhalchandra D. Thatte \\
  \small Department of Statistics, University of Oxford,\\[-0.8ex]
  \small Oxford OX1 3TG, United Kingdom\\
  \small \texttt{bdthatte@gmail.com} \\}

\begin{document}
\maketitle

\begin{abstract}

  Pedigrees are directed acyclic graphs that represent ancestral
  relationships between individuals in a population. Based on a
  schematic recombination process, we describe two simple Markov models
  for sequences evolving on pedigrees - Model R (recombinations without
  mutations) and Model RM (recombinations with mutations). For these
  models, we ask an identifiability question: is it possible to
  construct a pedigree from the joint probability distribution of extant
  sequences? We present partial identifiability results for general
  pedigrees: we show that when the crossover probabilities are
  sufficiently small, certain spanning subgraph sequences can be counted
  from the joint distribution of extant sequences. We demonstrate how
  pedigrees that earlier seemed difficult to distinguish are
  distinguished by counting their spanning subgraph sequences.

\end{abstract}



\section{Introduction}\label{intro}
Phylogenetics is a study of how species are related to each
other. Evolutionary relationships are most conveniently represented by
rooted leaf-labelled trees, where the leaves represent extant species
and the root represents their most recent common ancestral
species. Similarly other internal vertices of evolutionary trees
correspond to extinct ancestral species.

The arrival of DNA and protein sequence data in the last forty years led
to explosive growth of phylogenetics. Many of the modern phylogenetic
methods consider sequence data under probabilistic models of sequence
evolution. For such models to be useful for phylogenetic inference, it
is important to establish their {\em identifiability} (i.e., to show
that nonisomorphic trees or different model parameters cannot induce the
same distribution on the sequences at the leaves under a given model of
sequence evolution). Mathematical theory of phylogenetic trees,
especially probabilistic models of sequence evolution, the associated
questions of identifiability and statistical consistency (especially of
the maximum likelihood methods) have been extensively studied
\cite{ss03,gs07}, giving a firm statistical foundation to the study of
phylogenetic trees.

While phylogenetic trees represent relationships between species,
population pedigrees represent how individuals within a population are
related to each other. Communities all over the world have long been
curious about knowing their ancestral histories, and have often kept
detailed records of their family trees. In fact this curiosity goes back
much further in the past than the interest in constructing evolutionary
relationships between species. An example of a fairly detailed record of
family histories is the Icelandic database \'Islendingab\'o (The Book of
Icelanders \url{http://www.islendingabok.is}) of genealogical records
that covers almost the whole Icelandic population and goes back to
nearly 1200 years. Such ancestral histories are often compiled from a
variety of sources such as church records, birth and death records,
obituaries etc. that are prone to ambiguities or missing data beyond a
few generations in the past.

In the last several years large amounts of data on intra-population
genetic variation have been recorded. For example, the Icelandic
biomedical company deCode Genetics has compiled genomic, genealogical
and health data of more than 100000 individuals (which is a significant
proportion of the current Icelandic population).  Such data offer
promising opportunity to cross check and resolve ambiguities in historic
genealogical records besides being useful for other studies such as, for
example, genetic factors associated with medical conditions. Therefore,
there is a renewed interest in accurately inferring pedigrees from
genomic data. The statistical and combinatorial foundation for studying
reconstruction problems for pedigrees has not been as developed as in
phylogenetics. A purpose of this paper is to continue our earlier
attempts to develop such a foundation for the problem of reconstructing
pedigrees from observations (sequences) on extant individuals.

To develop such a foundation, we need to establish results along the
following lines: developing a biologically realistic model for sequences
undergoing mutations and recombinations, and identifiability results for
such a model; statistical consistency results; and finally results that
give estimates for the amount of genomic data necessary to reliably
construct pedigrees. This paper is mainly about identifiability
questions.

In the rest of this section, we discuss the above theoretical motivation
in more detail. We begin by informally sketching some well known
reconstruction and identifiability results for phylogenetic trees that
not only motivate the work in this paper but also are crucially useful
in the proof of the main identifiability result of this paper.

\

\noindent {\em Results of Zarecki{\u\i} and Buneman.}
It was shown in \cite{zar65} that a leaf-labelled tree can be uniquely
constructed from the pairwise distances between its leaves. The result
was strengthened slightly as follows \cite{buneman71}. Suppose that $f$
is an {\em additive function} on the family of subsets of cardinality 2
of the vertex set of a leaf-labelled tree. Here {\em additive} means
that for any two vertices $r$ and $s$, we have $f(\{r,s\}) = \sum f(e)$,
where the summation is over all edges $e$ on the (unique) path from $r$ to
$s$. Buneman showed that knowing $f$ on all pairs of leaves of a
leaf-labelled tree without vertices of degree 2 is sufficient to
uniquely construct the tree and the function. It is not surprising that
these results are quite useful in phylogenetics, where observations on
extant species (leaves of an evolutionary tree) are used to infer a
suitably defined distance (or an additive function) between pairs of
species, and then their phylogenetic tree is constructed uniquely.

\

\noindent {\em Results of Steel and Chang.} Now suppose the evolutionary
process on a (rooted) tree is modelled as follows: first the root is
assigned a random state from a finite alphabet $\Sigma$ (e.g., $\Sigma$
may be \{A,T,G,C\} in the case of DNA sequences). Each state is assumed
to have a nonzero probability of being assigned to the root. Each edge
of the tree has associated with it a $|\Sigma| \times |\Sigma|$ matrix
of substitution probabilities. These substitution probabilities
determine how the vertex-states evolve away from the root, and induce a
distribution of states at the leaves of the tree. The model was
formulated in \cite{steel94} as described above, while a slightly more
general formulation in terms of Markov random fields on unrooted trees
was given in \cite{chang96}. It was independently shown by Chang and
Steel that when the matrices defining the substitution probabilities
satisfy certain mild conditions, the (unrooted) tree can be uniquely
recovered from the joint distribution of states at the leaves of the
tree. In particular they showed that the negative logarithm of the
determinant of the matrix of substitution probabilities between pairs of
vertices is an additive function on the pairs of vertices, and can be
computed from the probability distribution on extant sequences. It was
further proved in \cite{chang96} that the substitution matrices are
also identifiable from the marginal distributions on triples of leaves
of the tree. Special cases of these results for models more commonly
used in phylogenetics were known in the phylogenetics literature
earlier.

\

How far can we generalise such results if the underlying structure is
more general than a tree? A recent result in this direction is due to
\cite{bms08} where it is shown that under some mild non-degeneracy
conditions the dependency structure of a Markov random field can be
obtained from sufficiently many independent samples.

In this paper we present simple models for recombinations and mutations
for population pedigrees, and generalise phylogenetic identifiability
results for them. One difference between reconstructing Markov random
fields and reconstructing pedigrees is that for pedigrees we have
observations only on the extant individuals (e.g., DNA sequences derived
from living individuals). Moreover, in the problem of reconstructing
pedigrees, the samples of data (e.g., columns in a sequence alignment)
are not {\em i.i.d.} (independent and identically distributed) as a
result of recombinations.

In \cite{sh06,thatte18}, we studied some purely combinatorial
reconstruction problems motivated by Zarecki{\u\i}'s result, for
example, the problem of reconstructing a pedigree from the pairwise
distances between its extant individuals or from its subpedigrees
(pedigrees of subsets of the extant population). In \cite{thatte18}, we
showed that a pedigree cannot in general be reconstructed from the
collection of its proper subpedigrees. Such a result implies that
knowing pairwise distances between extant vertices is in general not
enough to reconstruct pedigrees.

In \cite{thatte16}, we considered models for sequences evolving on
pedigrees, and showed that for certain simple Markov models, pedigrees
are not identifiable from the distribution of observed states at extant
vertices. We did construct examples of processes for which pedigrees
could be proved to be identifiable, but the processes lacked the
Markovian property, which informally states that the state observed at a
vertex depends only on the states of its parents. Moreover, it seems
that pedigrees that are difficult to reconstruct in a purely
combinatorial framework (e.g., from pairwise distances between extant
vertices or from subpedigrees) are also likely to be difficult to
reconstruct in a stochastic framework. For example, if a pedigree cannot
be reconstructed from its proper subpedigrees, then the marginal
distributions of extant sequences on proper subsets of the extant
population might be insufficient to uniquely recover the pedigree. On
the other hand, negative results in a combinatorial setting may not
imply non-identifiability in a stochastic framework. It is therefore
important to study combinatorial reconstruction problems
(e.g., classification of pedigrees that may be difficult to reconstruct
combinatorially), stochastic identifiability problems for idealised
Markov models of recombination and mutation, and relationships between
these problems.

Reconstruction problems of purely combinatorial nature are well known to
combinatorialists, the foremost among such problems being the vertex
reconstruction conjecture \cite{ulam60}. The conjecture states that all
simple undirected unlabelled graphs can be constructed from their
collection of unlabelled induced subgraphs. Combinatorial reconstruction
problems have also been studied in phylogenetics, for example, problems
of reconstructing phylogenetic trees from subtrees \cite{be04}.

Steel and Chang proved their phylogenetic identifiability results in two
parts: computation of the additive `log determinant' function on pairs
of leaves from the joint probability distribution on leaf states, and
the combinatorial problem of reconstructing a tree from the additive
function, which had been solved by Buneman and Zarecki{\u\i}. Similarly
the problem of reconstructing pedigrees under a recombination-mutation
model may be solved in two parts: in the first part, we would like to
reduce the identifiability question to an appropriate combinatorial
reconstruction problem, and then in the second part we would like to
show that the combinatorial reconstruction problem has a unique
solution. To ensure the uniqueness of reconstruction, we will have to
compute sufficiently strong combinatorial invariants of the pedigree
from the joint distribution of extant sequences.

Although we noted that distances between extant vertices in a pedigree
are not sufficient to reconstruct a pedigree, we sketch a heuristic
argument given in \cite{ts08} that shows how the distances between
extant vertices in a discrete generation pedigree may be obtained from
sequence data. Suppose $a$ and $b$ are two extant individuals in a
pedigree. Suppose further that in the pedigree there are $n := n(k:
a,b)$ pairs of paths such that one of the paths in each pair ends on $a$
and the other ends on $b$, and the two paths in each pair start at a
common ancestor of $a$ and $b$ in the $k$-th generation, and the two
paths in a pair do not share any other vertex. Now if sufficiently many
short recombination-free homologous segments of DNA of $a$ and $b$ are
compared, then we would expect about $n/2^{2k}$ of them to have a common
ancestor in the $k$-th generation. Thus it may be possible to estimate
$n(k:a,b)$. Such calculations are theoretically possible for small
values of $k$ assuming the population is large and the sequences are
long. They then tried to use the numbers $n(k:a,b)$ for all pairs
$\{a,b\}$ for all $k$ to construct the pedigree. They computationally
found many pairs of non-isomorphic pedigrees that have the same number
of pairs of paths of each length. One such example is the pair of
pedigrees shown in Figure~\ref{fig-hap}, which was also mentioned in
\cite{sh06}.

But a more detailed analysis of sequence similarities (between multiple
sequences, if required) under simple recombination and mutation models
may give us more information than just pairwise distances between living
individuals. In the main theorem of this paper
(Theorem~\ref{thm-paths}), we show that the joint distribution on extant
sequences determines a class of combinatorial invariants (e.g., certain
types of subgraph sequences) that supersedes pairwise distances between
extant sequences and subtrees (genealogical trees) in a pedigree.  We
then show that pedigrees, such as the ones in Figure~\ref{fig-hap}, that
earlier seemed difficult to distinguish due to their combinatorial
similarities (including the non-reconstructible pedigrees constructed in
\cite{thatte18}) are distinguished by the class of invariants.

\

This paper is organised as follows. In Section~\ref{sec-prelim}, we
define pedigrees, alignments and subgraph sequences. In
Section~\ref{sec-models}, we give a schematic description of the
recombination process, and formalise three models for sequences: Model R
(a model in which there are recombinations but no mutations), Model RM
(a model in which there are recombinations and mutations), and Model M
(a model of mutations for sequences evolving on trees). We then
formulate identifiability problems for these models. In
Section~\ref{sec-model-r}, we analyse pedigrees with two generations
under Model R. In Section~\ref{sec-model-rm}, we prove the main theorem
(Theorem~\ref{thm-paths}) and demonstrate its applications. In the last
section, we discuss a few open questions. A section on nomenclature
follows the references where all symbols appearing in the paper and the
number of the page on which they appear first are listed.
 
\section{Pedigrees, alignments and subgraph sequences} \label{sec-prelim}

We use the following notation for number systems and their subsets:
$\zz$ for the set of integers, \nomenclature{$\zz$}{the set of integers}
$\zz_+$ for the set of positive integers, \nomenclature{$\zz_+$}{the set
  of positive integers} $\nn$ for natural numbers,
\nomenclature{$\nn$}{the set of natural numbers} $[m]$ for the set
$\{1,2,\ldots,m\}$. \nomenclature{$[m]$}{the set $\{1,2,\ldots, m\}$}
Depending on the context, we write $[a,b]$ for the set of integers
$\{a,a+1, \ldots, b\}$ or real numbers $a\leq x \leq b$ (and similarly
$(a,b)$, $[a,b)$ and $(a,b]$ for open or half-open intervals in integers
and reals). \nomenclature{$[a,b], [a,b), \ldots$}{intervals in integers
  and reals} The set of all $k$-tuples of elements of a set $S$ is
written as $S^k := \{(s_1,s_2,\ldots, s_k) \mid s_i \in S, i \in [k]\}$.
\nomenclature{$S^k $}{the set of $k$-tuples of elements of a set $S$}
The set of all functions from $X$ to $S$ is written as $S^X :=
\{f:X\rightarrow S\}$.

\nomenclature{$S^X$}{the set of all functions from $X$ to $S$}

Next we introduce some graph theoretic notation. The vertex set and the
edge set of a graph $G$ are denoted by $V(G)$ and $E(G)$, respectively,
and their cardinalities by $v(G)$ and $e(G)$, respectively. The
in-degree and the out-degree of a vertex $u$ in a directed graph are
denoted by $d^-(u)$ and $d^+(u)$, respectively. The degree of a vertex
$u$ in an undirected graph (or the total degree in a directed graph) is
denoted by $d(u)$. An arc from $u$ to $v$ in a directed graph and also
an edge between $u$ and $v$ in an undirected graph is written as $uv$,
and it will be understood from the context whether $uv$ is meant to be a
(directed) arc or an (undirected) edge. When any two objects $G_1$ and
$G_2$ are isomorphic, we write $G_1 \cong G_2$. The isomorphism class of
an object $G$ is written as $\lVert G \rVert$. For a collection $\mcg$
of labelled objects, we write $\lVert \mcg \rVert$ for the set of
isomorphism classes of objects in $\mcg$. Let $G$ and $H$ be two
directed or undirected graphs. We write $G \leq H$ (or $H \geq G$) if
$G$ is isomorphic to a subgraph of $H$, and this notation may be used
when $G$ or $H$ is unlabelled (i.e., they are just isomorphism
classes). We write $G \subseteq H$ (or $H \supseteq G$) when a labelled
graph $G$ is a subgraph (or a supergraph) of a labelled graph $H$.

\nomenclature{$V(G), E(G)$}{- vertex and edge sets of a graph,
  respectively}

\nomenclature{$v(G), e(G)$}{- cardinalities of vertex and edge sets of a
  graph, respectively}

\nomenclature{$G \cong H$}{- $G$ and $H$ are isomorphic}

\nomenclature{$G \leq H$ or $H \geq G$}{- when used for graphs (or
  isomorphism classes of graphs) $G$ and $H$, it means $G$ is isomorphic
  to a subgraph of $H$}

\nomenclature{$G \subseteq H$ or $H \supseteq G$}{- when used for labelled
  graphs $G$ and $H$, it means $G$ is a subgraph of $H$}

\nomenclature{$\lVert S \rVert$}{- isomorphism class of an object; if $S$
  is a class of objects, then it is the set of isomorphism classes of
  objects in the class}

\nomenclature{$d^-(u), d^+(u), d(u)$}{- in-degree, out-degree, degree (or
  total degree) of a vertex $u$}

\begin{defn}[{\bf General pedigrees}] \label{defn-general} A {\em
    general pedigree} $P(X,Y,U,E)$ of a set $X$ is a directed acyclic
  graph on a vertex set $U\supseteq X\cup Y$ and a set of arcs $E$ such
  that each vertex has in-degree 0 or 2. The set $X$ is the set of
  vertices with out-degree 0. The set $Y$ is the set of vertices with
  in-degree 0. The vertices in $X$ are called the {\em extant vertices}
  (or the extant individuals in the population). The vertices in $Y$ are
  called the {\em founder vertices} (or the founders of the
  population). The {\em order} of the pedigree is $|X|$. The {\em depth}
  of a pedigree is the length of (i.e., the number of arcs in) a longest
  path in the pedigree. Two pedigrees $P(X,Y,U,E)$ and $Q(X,Z,V,F)$ are
  said to be {\em isomorphic} if there is a one-one map
  $\pi:U\rightarrow V$ such that $uv$ is an arc in $P$ if and only if
  $\pi(u)\pi(v)$ is an arc in $Q$, and $\pi(x) = x$ for all $x \in
  X$. We denote the natural partial order on $U$ by $\leq $, i.e., $v
  \leq u$ if there is a directed path from $u$ to $v$ (or $u = v$).
\end{defn}

\nomenclature{$P$, $Q$, $P(X,Y,U,E)$, ...}{pedigrees}

\nomenclature{$X$}{the set of extant vertices of a pedigree}

\nomenclature{$\cong$}{isomorphism between pedigrees, $X$-forests,
  graphs, etc.}

\nomenclature{$u \leq v$ - }{(for vertices $u$ and $v$ in a pedigree)
  there is a directed path from $v$ to $u$}

We define isomorphism only between pedigrees of the same set of extant
individuals and require that it fixes all extant vertices because
(informally speaking) we would like to treat extant vertices to be
labelled and other vertices to be unlabelled. Throughout this paper, we
will assume that all pedigrees have $X$ as their set of extant vertices.

\begin{defn}[{\bf Diploid pedigrees}]
  \label{defn-diploid}
  Let $P(X,Y,U,E)$ be a pedigree. Suppose that $U$ can be partitioned
  into unordered pairs of vertices such that the following conditions
  hold: the extant vertices are paired with extant vertices and the
  founder vertices are paired with founder vertices; two non-extant
  vertices $v$ and $w$ are paired if and only if there are arcs $vu$ and
  $wu$ in $E$; no two paired vertices have a common parent. Then $P$
  together with one such pairing is called a {\em diploid pedigree}.
\end{defn}

Thus any general pedigree is a {\em haploid pedigree}; a diploid
pedigree is a haploid pedigree with a pairing of its vertices (although
not all pedigrees admit such a pairing). The pairing in a diploid
pedigree is completely determined by the pairing of its extant vertices
since all other pairs are determined by the definition. An advantage of
a purely combinatorial definition of a diploid pedigree is that we can
now assign just one sequence to each vertex. From this point of view,
haploid pedigrees are pedigrees of sequences, not of
individuals. Pedigrees of individuals are diploid pedigrees obtained by
pairing sequences in a haploid pedigree. Figure~\ref{fig-d2h}
illustrates this point of view.  When vertices (sequences) $A_i$ and
$B_i$ are paired, $A_j$ and $B_j$ must be paired, since they are parents
of $A_i$; similarly, $A_k$ and $B_k$ must be paired.

\begin{figure}[ht]
\begin{center}
  \includegraphics[width=118mm]{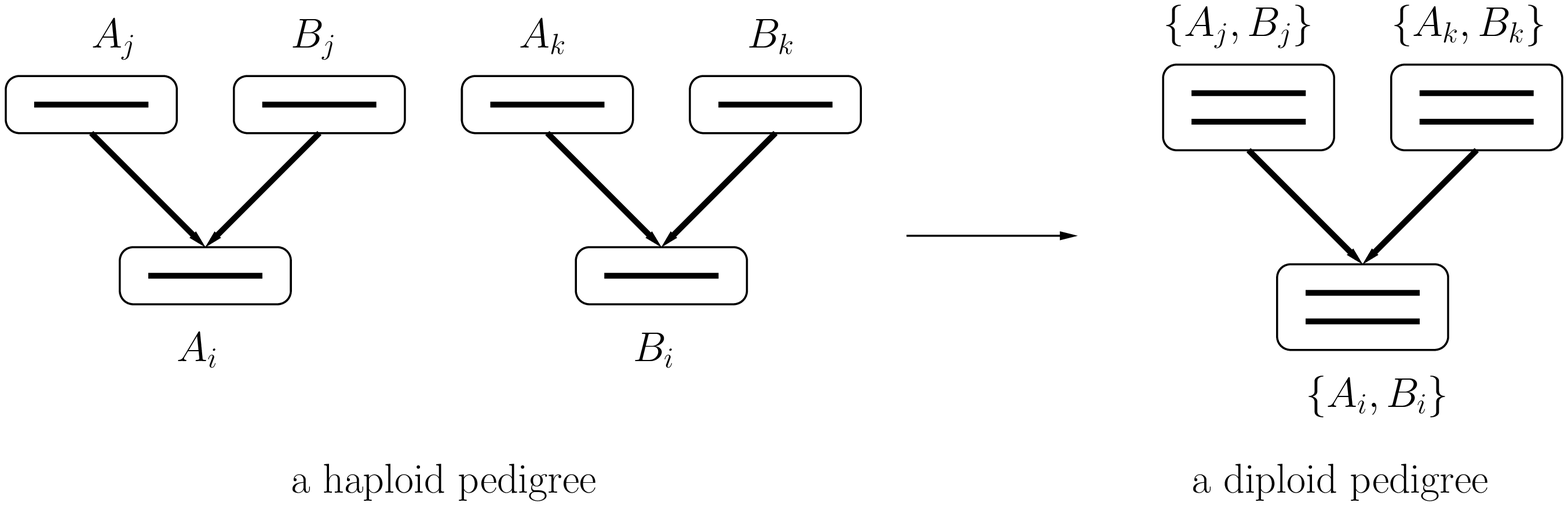}
\end{center}
  \caption[]{Pairing vertices (sequences) of a haploid pedigree}
  \label{fig-d2h}
\end{figure}

\begin{defn}[{\bf Sequence alignments and characters}]
  \label{defn-char} Let $\Sigma $ be a finite set (called an {\em
    alphabet}) and $U$ a finite set. A {\em character} on $U$ is a map
  $C:U \rightarrow \Sigma$. For $L \in \nn$, an {\em alignment} of
  length $L$ on $U$ is a map $A: U \rightarrow \Sigma^L$. Equivalently,
  an alignment is an $L$-tuple $(C_1, C_2, \ldots, C_L)$ of characters
  on $U$, or a two dimensional array of symbols from $\Sigma$, with
  $|U|$ rows and $L$ columns. The rows of the array are called {\em
    sequences}, and are written as $A(i)$, $i = 1$ to $|U|$. Individual
  entries in the array are written as $A(i,j)$, where $i = 1$ to $|U|$
  and $j=1$ to $L$. The columns of an alignment are called {\em sites}.
\end{defn}

\nomenclature{$\Sigma$}{finite alphabet}

\nomenclature{$\Sigma^{XL}$}{the set of alignments of length $L$ on $X$}

\nomenclature{$C$, $C_i$, ...}{characters on $X$, i.e., maps $C:X
  \rightarrow \Sigma$}

\nomenclature{$A,A_i,\ldots$ }{alignments on $X$, i.e., maps from $X$ to
  $\Sigma^L$ or elements of $\Sigma^{XL}$ }

\nomenclature{$\Sigma^X$}{the set of site patterns on $X$}

Usually $U$ will be the set of vertices of a pedigree, and we will be
interested in alignments restricted to the set $X$ of extant vertices of
the pedigree. The space of characters on $X$ is $\Sigma^X$, which is
also referred to as the space of {\em site patterns}. The set of
alignments of length $L$ on a set $X$ is $(\Sigma^{X})^L$, which we
write as $\Sigma^{XL}$.

Let $P(X,Y,U,E)$ be a pedigree and let $A \in \Sigma^{UL}$ be an
alignment on $P$. If the sequences in $A$ have evolved under some
process of recombination or mutation, then (regardless of the details of
the model of recombination or mutation) we may suppose that each site in
each sequence is inherited only from one of the two parent
sequences. Therefore, for each $u \in U$ and each $j \in [L]$, there is
a unique directed path $P_{uj}$ from some founder vertex $y_{uj}$ to $u$
that defines the genetic ancestry of $A(u,j)$. Therefore, each site $j$
has associated with it a spanning forest $G_j := \cup_{u\in U}P_{uj}$,
and each alignment of length $L$ has an underlying (usually unknown)
sequence $(G_j,j = 1,2,\ldots, L)$ of spanning forests. Similarly, each
site $j$ has associated with it a {\em directed subforest} defined by
$T_j := \cup_{x\in X}P_{xj}$, and we have a {\em directed subforest
  sequence} $(T_j,j = 1,2,\ldots, L)$ of the alignment. Here we define
these notions purely graph theoretically (without reference to
alignments or models).

\begin{defn}\label{defn-xtree}
  Let $X$ be a finite set. A {\em directed $X$-forest} $T$ is a directed
  forest that satisfies the following conditions: the set of vertices
  with out-degree 0 is $X$, and is called the {\em leaf set} of $T$;
  each component has a single vertex of in-degree 0, called the {\em
    root vertex} of the component, and all other vertices have in-degree
  1; all arcs are directed away from the root vertices. An {\em
    undirected $X$-forest} is an unrooted forest with leaf set
  $X$. Suppose $T$ is a directed $X$-forest. It induces a natural
  partition of $X$ into {\em maximal clusters} $X_1, X_2, \ldots, X_k$
  such that vertices in cluster $X_i$ have a unique most recent common
  ancestor (MRCA) $u_i$ in $T$. We construct a subgraph of $T$ induced
  by $u_i, i \in [k]$ and their descendants in $T$, and then replace its
  (directed) arcs by (undirected) edges. The resulting undirected
  unrooted graph is called the {\em undirected $X$-forest} of $T$,
  written as $T_u(T)$. It is the maximal undirected $X$-forest contained
  in $T$.
\end{defn}

In the above definition, the term {\em $X$-forest} is meant to be
analogous to the term {\em $X$-tree} that is commonly used in
phylogenetics \cite{ss03}. In this paper we will adapt some of the
phylogenetic identifiability results for undirected $X$-forests that
appear as undirected graphs underlying subgraphs of
pedigrees. Therefore, in the following definition, we specialise the
terms for pedigrees.

\begin{defn}\label{defn-subtree}
  Let $P$ be a pedigree.  A {\em spanning forest} of $P$ is a spanning
  subgraph $G$ of $P$ such that the in-degree of each vertex in $G$ is 1
  unless it is a founder vertex in $P$.  A {\em directed $X$-forest} of
  $P$ is a subgraph $T$ of $P$ such that $T$ is a directed $X$-forest
  and the root vertex of each component of $T$ is a founder vertex in
  $P$.  Each spanning forest $G$ in $P$ contains a unique directed
  $X$-forest of $P$, and we denote it by $T_d(G)$. An {\em undirected
    $X$-forest} in $P$ is the unique undirected $X$-forest in any
  directed $X$-forest in $P$. Each spanning forest $G$ in $P$ contains a
  unique undirected $X$-forest of $P$, and we denote it by $T_u(G)$ .
\end{defn}

\nomenclature{$T,T_i, \ldots$}{- $X$-forests in a pedigree or $X$-forests}

\nomenclature{$T_d(G)$}{- the unique directed $X$-forest in a spanning
  forest $G$ in a pedigree}

\nomenclature{$T_u(G)$}{- the unique undirected $X$-forest in a spanning
  forest $G$ in a pedigree}

\nomenclature{$G,G_i, \ldots$}{- spanning forests in a pedigree}

\nomenclature{$\mathbf{G}:=(G_1,G_2,\ldots, G_m)$}{- a spanning forest
  sequence of length $m$} 

\nomenclature{$\mathbf{T}:=(T_1,T_2,\ldots, T_m)$}{- an $X$-forest
  sequence of length $m$}

Note that we use the term {\em spanning forest} in a specific sense: a
spanning forest is not any spanning forest in the graph theoretic
sense. We illustrate these terms in Figure~\ref{fig-subtrees}, which
shows a pedigree and a spanning forest $G$ with $E(G) =
\{da,a1,a2,eb,b3,fc\}$ (shown by bold arcs). In this example, the unique
directed $X$-forest $T_d(G)$ in $G$ has the arc set
$\{da,a1,a2,eb,b3\}$; its clusters are $X_1 = \{1,2\}$ and $X_2 =
\{3\}$, and the root vertices of its components are $d$ and $e$. The
unique undirected $X$-forest $T_u(G)$ in $G$ consists of vertex set
$\{1,2,3,a\}$ and edge set $\{a1,a2\}$. Note that vertex 3 is isolated
in $T_u(G)$ since it is the MRCA of its cluster, but we need it in our
analysis.

\begin{figure}[ht]
\begin{center}
  \includegraphics[width=50mm]{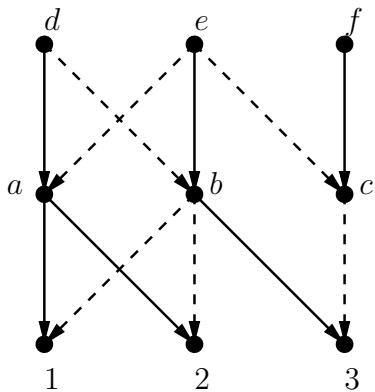}
\end{center}
\caption[]{A pedigree and a spanning forest, which is shown by bold
  arcs}
  \label{fig-subtrees}
\end{figure}

\begin{defn}
  Two (directed or undirected) $X$-forests $T$ and $T^{\prime}$ are said
  to be {\em isomorphic} (written $T\simeq T^{\prime}$) if there is a
  graph theoretic isomorphism $\pi$ from $T$ to $T^{\prime}$ such that
  $\pi(x) = x$ for all $x$ in $X$. The isomorphism class of a directed
  or an undirected $X$-forest $T$, denoted by $\lVert T \rVert$, is the
  set of all $X$-forests $T^{\prime}$ that are isomorphic to $T$.
\end{defn}
The set of all directed $X$-forests of $P$ is denoted by $\mct_P$. The
set of isomorphism classes of directed $X$-forests of $P$ (or the set of
{\em distinct} directed $X$-forests in $P$) is denoted by $\lVert
\mct_P\rVert$. The set of spanning forests of a pedigree $P$ is denoted
by $\mcg_P$. The set of undirected $X$-forests of $P$ is denoted by
$\mcu_P$. The set of isomorphism classes of undirected $X$-forests of
$P$ (or the set of {\em distinct} undirected $X$-forests in $P$) is
denoted by $\lVert \mcu_P\rVert$.

\nomenclature{$\lVert T \rVert$}{the isomorphism class of an $X$-forest
  $T$}

\nomenclature{$\mct_P, \mcu_P$}{the sets of directed and undirected
  $X$-forests in a pedigree $P$, respectively}

\nomenclature{$\lVert \mct_P \rVert, \lVert \mcu_P \rVert$}{the sets of
  isomorphism classes of directed and undirected $X$-forests (or
  distinct directed or undirected $X$-forests) in a pedigree $P$,
  respectively}

\nomenclature{$\mcg_P$}{the set of spanning forests in a pedigree $P$}

\begin{prop}\label{prop-spanning-forests}
  Each spanning forest of a pedigree $P$ has $e(P)^/2$ arcs and contains
  exactly one directed $X$-forest. There are $2^{e(P)/2}$ spanning
  forests. Each directed $X$-forest $T$ is contained in
  $2^{(e(P)-2e(T))/2}$ spanning forests.
\end{prop}

\begin{proof} Each spanning forest of $P$ is obtained by selecting one
  of the two arcs that point to each non-founder vertex. The unique
  directed $X$-forest $T$ in a spanning forest $G$ is the subgraph of
  $G$ spanned by vertices in $G$ that are ancestral in $G$ to the extant
  vertices. Finally, $e(P) - 2e(T)$ is the number of arcs not pointing
  to vertices in $T$, therefore, $(e(P) - 2e(T))/2$ is the number of
  non-founder vertices outside $T$, at each of which we can choose one
  of the two incoming arcs to construct a spanning forest containing
  $T$.
\end{proof}

In any model of recombination, the number of recombination events is
determined by the spanning forest sequence underlying an alignment, but
we define it for all spanning forest sequences without reference to
alignments or models.

\begin{defn}\label{def-rg-sg}
  Let $P$ be a pedigree and let $\mathbf{G} := (G_i, i = 1,2,\ldots, L)$
  be a spanning forest sequence in $P$. If $vu$ and $wu$ are distinct
  arcs in $P$ such that $vu$ is in $G_i$ and $wu$ is in $G_{i+1}$, then
  we say that a recombination has occurred at site $i$ at vertex $u$ or
  that $i$ is a recombining site. If $vu$ is an arc in $G_i$ and
  $G_{i+1}$ then we say that there is no recombination at $u$ at site
  $i$. We define the number of recombinations in $\mathbf{G}$ to be
  \[
  r(\mathbf{G}) := \sum\limits_{i=1}^{L-1}|E(G_{i+1})\bigtriangleup
  E(G_{i})|/2,
  \] 
  where $|E(G_{i+1})\bigtriangleup E(G_{i})|/2$ is the number of
  recombinations separating $G_i$ and $G_{i+1}$. The number of points of
  no recombination is
  \[
  s(\mathbf{G}) := \sum\limits_{i=1}^{L-1}|E(G_{i+1})\cap E(G_{i})|.
  \]
  The directed $X$-forest sequence of $\mathbf{G}$ is the sequence
  $\mathbf{T_d} := (T_d(G_i), i = 1,2,\ldots, L)$, and the undirected
  $X$-forest sequence of $\mathbf{G}$ is the sequence $\mathbf{T_u} :=
  (T_u(G_i), i = 1,2,\ldots, L)$.
\end{defn}

\nomenclature{$r(\mathbf{G})$}{the number of recombinations in
  $\mathbf{G}$, see Definition~\ref{def-rg-sg}}

\nomenclature{$s(\mathbf{G})$}{the number of points of no recombination in
  $\mathbf{G}$, see Definition~\ref{def-rg-sg}}

\section{Models R and RM, and identifiability problems}
\label{sec-models}
We assume that in any reasonable model of sequence evolution, sequences
are first assigned to the founder vertices, and then subsequent
generations of individuals inherit their sequences from their parents'
sequences subject to recombinations and mutations. We are then
interested in the following types of {\em identifiability} questions.

\begin{problem} \label{problem-main} Suppose sequences of equal length
  over a finite alphabet are assigned to the founder vertices of a
  pedigree. The sequences then evolve on the pedigree undergoing
  recombinations and mutations, giving a probability distribution on the
  space of alignments on the set of extant vertices. Can we determine
  the pedigree uniquely (i.e., up to isomorphism) - with or without the
  knowledge of the size or the depth of the pedigree or the various
  probability parameters defining the recombination and mutation
  processes, with or without restrictions such as discrete generations
  or constant population, and so on? In the case of diploid pedigrees,
  we will be given the distribution on alignments on the set of extant
  vertices along with a pairing of extant vertices.
\end{problem}

In this paper, we study the above types of questions under two simple
models of recombination and mutation. In {\em Model R}, we assume that
sequences evolve on a pedigrees under a process of recombinations
without mutations. In {\em Model RM}, we assume that sequences evolving
on a pedigree undergo recombinations and mutations. For convenience, we
also formalise the mutation part of Model RM for the spanning forests of
pedigrees and in general for directed $X$-forests, and call it {\em
  Model M}.

In all these models, we assume that first all the founder vertices of a
pedigree (or the root vertices of a spanning forest or a directed
$X$-forest) are assigned sequences. These sequences are independently
selected from a uniform distribution on $\Sigma^L$, where $\Sigma $ is a
known finite alphabet. Then the sequences evolve on the pedigree (or a
spanning forest or a directed $X$-forest) in a top-down manner, i.e., a
vertex is assigned a sequence only after its parents have been
assigned sequences.

We begin with a schematic description of the recombination process. Our
description is largely based on Chapter 12 of \cite{lange02}.
Figure~\ref{fig-meio} schematically shows the process of gamete (sperm
or egg) formation in eukaryotes. Initially there is a parent cell with
one pair of homologous non-sex chromosomes. Each chromosome is then
duplicated with the identical sister {\em chromatids} joined together at
the {\em centromere}, forming a {\em four-strand bundle}. Then the two
duplicated chromosomes exchange material between {\em chiasmata} ({\em
  recombination points}). In the diagram, there are three recombination
events. The first recombination is between strands 1 and 3 (counted from
top to bottom), the second is between strands 2 and 4, and the third is
between strands 2 and 3. The four chromatids after the exchange of
material are shown next. Then the cell undergoes two cell divisions to
create four haploid gametes, each receiving one of the four chromosomes.

As shown in the diagram, at each recombination point a crossover occurs
between one strand from the first pair and one strand from the second
pair. At each crossover, a strand from the first pair and a strand from
the second pair are chosen randomly with equal probability (independent
of other chiasma). This independence property is known as the lack of
{\em chromatid interference}.

Suppose that the locations of recombinations are modelled as a Poisson
point process along the sequence (or on $[0, \infty)$) with the rate
  $\lambda$ or a Bernoulli process with probability $p$ (so that a
  crossover occurs after a site on a sequence with probability $p$
  independently of other sites or sequences). Since exactly two of the
  four gametes - one from the first pair and one from the second pair -
  inherit any recombination, any given gamete inherits a recombination
  with probability $1/2$. Therefore, for the sequence of any fixed
  gamete, the locations of recombinations are still modelled by a
  Poisson process, but with the rate $\lambda /2$, (or Bernoulli process
  with probability $p/2$). Therefore, a model may be formalised with
  just two parent sequences instead of four. A Poisson process for the
  locations of chiasmata was first proposed in \cite{haldane19}. Based
  on the above description, we formalise models R and RM, in which we
  assume that crossovers in a finite sequence occur according to a
  Bernoulli process.

\begin{figure}[ht]
  \begin{center}
    \includegraphics[scale=0.6]{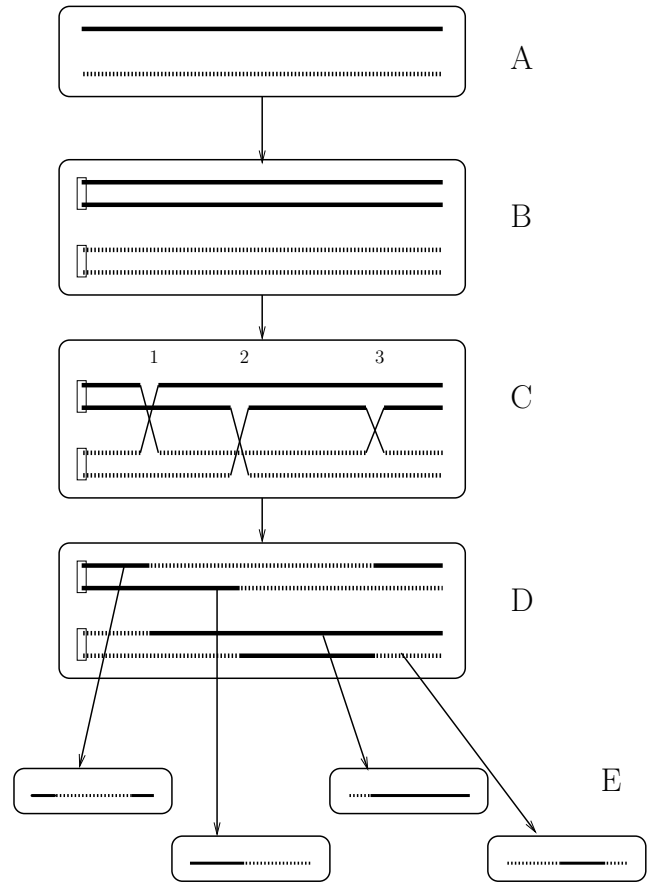}
  \end{center}
  \caption[]{Schematic description of recombination for diploid cells -
    (A) two homologous chromosomes in a parent cell (B) each chromosome
    is duplicated and a 4-strand bundle is formed (C) the sister
    chromatids of the first chromosome exchange material with the sister
    chromatids of the second chromosome between recombination points
    1,2,3 (D) four chromatids after the exchange of material (E) the
    four strands are inherited by four gametes}
  \label{fig-meio}
\end{figure}

\

\noindent {\em Model R:} Consider three sequences $A(i)$ of length $L$
over alphabet $\Sigma$, where $i \in \{u,v,w\}$ and $v$ and $w$ are
parents of $u$. The sequence $A(u)$ is obtained by recombining sequences
$A(v)$ and $A(w)$ as follows. Let $X_1,X_2,\ldots $ be a Markov chain on
the state space $\{v,w\}$ with transition probabilities $p_{ij} = p$ if
$i \neq j$ for $i,j \in \{v,w\}$, and $\pr\{X_1 = v\} = \pr\{X_1 = w\} =
1/2$. Then, for $k = 1,2,\ldots, L$, $A(u,k) \leftarrow A(i,k)$ if $X_k
= i$. Thus $X_{k+1} \neq X_k$ indicates a crossover from one sequence to
the other. We refer to this model as {\em Model R}.

\

\noindent {\em Model RM:} Consider three sequences $A(i)$ of length $L$
over alphabet $\Sigma$, where $i \in \{u,v,w\}$ and $v$ and $w$ are
parents of $u$.  The sequence $A(u)$ is obtained from the sequences
$A(v)$ and $A(w)$ by a process of recombinations and mutations as
follows.  Let $X_1,X_2,\ldots $ be a Markov chain on the state space
$\{v,w\}$ with transition probabilities $p_{ij} = p$ if $i \neq j$ for
$i,j \in \{v,w\}$, and $\pr\{X_1 = v\} = \pr\{X_1 = w\} = 1/2$. Then if
$X_k = i$ and $A(i,k) = r$, then $A(u,k)$ is assigned $r$ with
probability $1-(\Sigma -1)\mu$, and $A(u,k)$ is assigned a state
different from $r$ with probability $(\Sigma -1)\mu$. When a state
different from $r$ is assigned to $A(u,k)$, each state in $\Sigma
\backslash \{r\}$ has equal probability $\mu$ of being assigned to
$A(u,k)$. We refer to this process as {\em Model RM}.

\

\noindent {\em Model M:} This process is defined for spanning forests of
a pedigree and directed $X$-forests. First each founder or the root
vertex in each component of a directed $X$-forest is assigned
independently and uniformly randomly a state from $\Sigma$. Suppose $j$
is the parent vertex of $i$ in a spanning forest of a pedigree or in a
directed $X$-forest. Let $A(i)$ and $A(j)$ be the sequences of $i$ and
$j$, respectively, both of equal length $L$ over alphabet $\Sigma$. Then
for each $k \in [L]$, $A(i,k)$ is assigned the same state as $A(j,k)$
with probability $1-(\Sigma -1)\mu$, and $A(i,k)$ is assigned a state
different from $A(j,k)$ with probability $(\Sigma -1)\mu$. When a state
different from $A(j,k)$ is assigned to $A(i,k)$, each state in $\Sigma
\backslash \{A(j,k)\}$ has equal probability $\mu$ of being assigned to
$A(i,k)$. We refer to this process as {\em Model M}.

Model M on a directed $X$-forest $T$ is equivalent to a similarly
formulated model on the undirected $X$-forest of $T$. We root each
component of the undirected $X$-forest arbitrarily, and assign to it a
state from $\Sigma$ uniformly randomly, independent of the roots of
other components. The state then evolves away from the root in each
component. If a component itself is an isolated vertex, it is simply
assigned a state from $\Sigma$ uniformly randomly. Since the mutation
model described here is reversible, the same distribution on the site
patterns is observed on $X$ in the undirected $X$-forest as in a
directed $X$-forest for a given $\mu$. Therefore, when we try to
construct a tree from the character distribution on its leaves, we
cannot construct the directed $X$-forest, but we can at best construct
the undirected $X$-forest in it. Therefore, we will consider Model M
only on undirected $X$-forests.

\nomenclature{$p$}{the crossover probability in models $R(p)$ and
  $RM(p,\mu)$} \nomenclature{$\mu$}{the substitution probability in
  models $RM(p,\mu)$ and $M(\mu)$} \

Thus Model RM may be thought of as a synthesis of Models R and M so that
the recombination-free segments of sequences evolving on a pedigree may
be examined under Model M with phylogenetic methods.

Let $P$ be a pedigree on $X$. For an alignment $A\in \Sigma^{XL}$, we
denote by $\pr\{A \cond P, RM(p,\mu)\}$ the probability that sequences
of length $L$ evolving on the pedigree $P$ under model $RM(p,\mu)$ give
an alignment $A$ on $X$. We use analogous notation when the model
$RM(p,\mu)$ is replaced by the model $R(p)$, or when the pedigree $P$ is
replaced by a directed or an undirected $X$-forest $T$ and the model
under consideration is the mutation model $M(\mu)$. We denote the
various probability spaces by $(\Sigma^{XL}: P, RM(p,\mu))$,
$(\Sigma^{XL}: P, R(p))$, $(\Sigma^{XL}: T, M(\mu))$, $(\Sigma^{X}: T,
M(\mu))$, and so on. For pedigrees $P$ and $Q$, we will write
$(\Sigma^{XL}: P, RM(p,\mu)) = (\Sigma^{XL}: Q, RM(p,\mu))$ when $\pr\{A
\cond P, RM(p,\mu)\} = \pr\{A \cond Q,RM(p,\mu)\}$ for all $A \in
\Sigma^{XL}$, and analogously for other models.
\nomenclature{$(\Sigma^{XL}:P,RM(p,\mu))$}{the space of alignments on
  the extant vertices of $P$ as a probability space under model
  $RM(p,\mu)$; other probability spaces are denoted analogously}

As in the case of alignments, we treat the spaces of spanning forest
sequences, directed $X$-forest sequences, and undirected $X$-forest
sequences as probability spaces, and denote them by $(\mcg^{L}:P,R(p))$,
$(\mct^{L}:P,R(p))$, and $(\mcu^{L}:P,R(p))$, respectively. The spanning
forest sequences, and directed and undirected $X$-forest sequences, (and
the corresponding sequences of isomorphism classes of spanning forests
and directed and undirected $X$-forests) are defined by only the
recombination events, therefore, the probability spaces are unchanged if
$R(p)$ is replaced by $RM(p,\mu)$.  For a spanning forest sequence
$\mathbf{G} := (G_1,G_2,\ldots, G_L)$, we will write $\pr\{\mathbf{G} \cond 
P,R(p)\}$ for the probability of $\mathbf{G} $ in the probability space
$(\mcg^{L}: P,R(p))$. We will use analogous notation for other sequences
and probability spaces.  Unless stated otherwise, when we will refer to
alignments or sequences of spanning forests or other objects, we will
mean alignments or sequences of spanning forests or other objects,
respectively, from the appropriate probability spaces that are clear in
the context.

\begin{defn}\label{def-identifiability}
  Nonisomorphic pedigrees $P$ and $Q$ in a class $\mcc$ are said to be
  distinguished from each other under model $RM(p,\mu)$ if
  $(\Sigma^{XL}: P, RM(p,\mu)) \neq (\Sigma^{XL}: Q, RM(p,\mu))$ for
  some $L$, (i.e., for some $L$, there exists $\mca \subseteq
  \Sigma^{XL}$ such that $\pr\{\mca \cond P,RM(p,\mu)\} \neq \pr\{\mca
  \cond Q,RM(p,\mu)\}$). A pedigree $P$ in a class $\mcc$ is said to be
  identifiable under model $RM(p,\mu)$ if it is distinguished from every
  other pedigree $Q$ in $\mcc$, (i.e., if there is a pedigree $Q$ in
  $\mcc$ such that $(\Sigma^{XL}: P, RM(p,\mu)) = (\Sigma^{XL}: Q,
  RM(p,\mu))$ for all $L\in \zz_+$, then $Q$ is isomorphic to $P$).
  Pedigrees in a class $\mcc$ are said to be identifiable under model
  $RM(p,\mu)$ if all pairs of pedigrees in $\mcc$ are distinguished from
  each other under model $RM(p,\mu)$.
\end{defn}

\nomenclature{$\pr\{.\}$}{the probability of an event}

\nomenclature{$\mca, \mca_{i}$}{subsets of an alignment space such as
  $\Sigma^{XL}$}

Similar terminology will be used for other models and for undirected
$X$-forests. Stronger notions of identifiability may be defined, and
correspondingly stronger variants of identifiability questions may be
asked. The above definitions assume the model parameters to be
fixed. But we may ask if there are nonisomorphic pedigrees $P$ and $Q$
and model parameters $p, p^\prime, \mu, \mu^\prime$ such that
$(\Sigma^{XL}: P, RM(p,\mu)) = (\Sigma^{XL}: Q, p^\prime,\mu^\prime)$
for all $L \in \zz_+$. Given the probability distribution a pedigree
induces on the space of alignments, we may ask if the pedigree can be
recognised to be in a class $\mcc$. We assume in all results in this
paper that the model parameters and the class $\mcc$ (typically defined
by the size of a pedigree) to be fixed (but possibly unknown).

\begin{rem}
  Pedigrees of order 1 are in general not identifiable since for all
  pedigrees of order 1, the extant sequence will be uniformly
  distributed. In fact for all pedigrees, all extant sequences will be
  uniformly distributed. But when there are more than one extant
  vertices, the joint distribution of extant sequences will contain some
  information about the pedigree on which they have evolved, because,
  for example, some vertices may have common ancestors, so their
  sequences will be correlated. We will therefore consider
  identifiability questions for pedigrees of order more than 1.
\end{rem}

The models R and RM on a pedigree $P$ with $e(P) = 2e$ arcs may be
interpreted as Hidden Markov Models with the set of hidden states
$\mcg_P$ and the set of observed states $\Sigma^X$. This is illustrated
in Figure~\ref{fig-chains}. The initial probability for each hidden
state is $1/2^{e}$. The probability of an observed state conditional on
a given hidden state $G \in \mcg_P$ can be easily computed (and actually
depends only on $T_u(G)$ and $\mu$). The probability of transition from
state $G_i$ to $G_j$ is given by
\begin{equation}\label{eq-gi-gj}
  \pr\{G_j\mid G_i\} = p^{|E(G_j)\bigtriangleup E(G_i)|/2}
  (1-p)^{|E(G_j)\cap E(G_i)|}
\end{equation}

In most contexts in which HMMs are used, one assumes that the set of
hidden states in known, but here we do not know the set $\mcg_P$. We
also note that the sequence of $X$-forests is not a Markov
chain. Consider directed $X$-forests $T_j$ and $T_{j+1}$ at sites $j$
and $j+1$, respectively. Suppose a vertex $u$ in $V(T_{j+1}) \backslash
V(T_j)$ has parents $v$ and $w$ in the pedigree. The probability that
the arc $vu$ is in $T_{j+1}$ depends on the history before the $j$-th
site. For example, if $vu$ was in $T_{j-1}$ then the probability that it
would also be in $T_{j+1}$ is $(1-p)^2 + p^2$. But if $vu$ was not in
$T_{j-1}$ the probability that it would be in $T_{j+1}$ is
$2p(1-p)$. Therefore, $T_1, T_2, \ldots $ is not a Markov chain, i.e.,
$\pr\{T_{j+1} \cond T_j, T_{j-1}, \ldots, T_1 , P,R(p)\}$ may not be the
same as $\pr\{T_{j+1} \cond T_j, T^{\prime}_{j-1}, \ldots, T^{\prime}_1,
P,R(p)\}$ if the sequences $T_{j-1}, \ldots, T_1$ and $T^{\prime}_{j-1},
\ldots, T^{\prime}_1$ are different. Therefore, we cannot interpret the
models as HMMs with directed or undirected $X$-forests as hidden states.

\begin{figure}[ht]
\begin{center}
  \includegraphics[width=118mm]{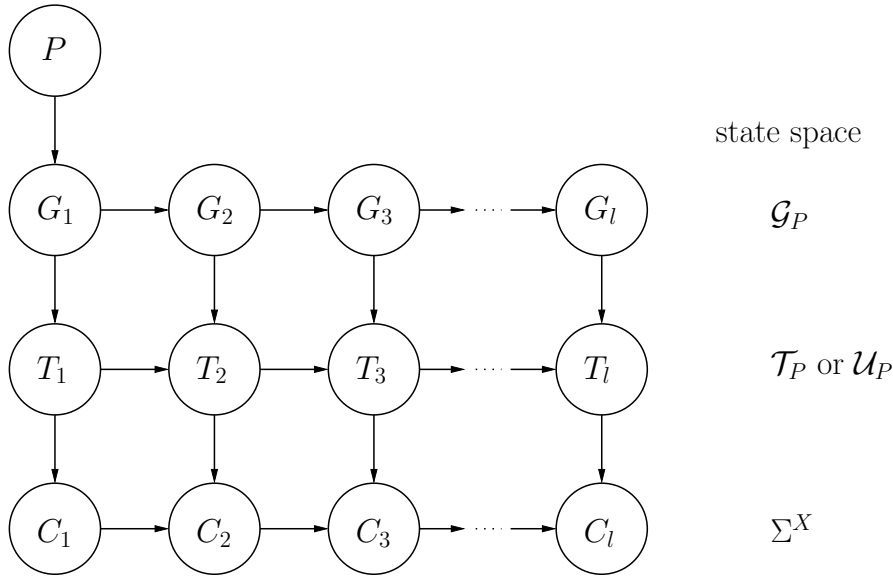}
\end{center}
\caption[]{An observed sequence $C_1, C_2, \ldots $ of characters at the
  extant vertices of a pedigree $P$. Among the intermediate chains, only
  $G_1,G_2, \ldots $ is a Markov chain.}
  \label{fig-chains}
\end{figure}

\section{Analysis of some examples under Model R}
\label{sec-model-r}

In this section, we analyse diploid pedigrees and certain haploid
pedigrees that obey many properties of diploid pedigrees under model
$R$. We show that, with some exceptions and mild conditions on $p$ and
$|\Sigma|$, diploid pedigrees of depth 2 are reconstructible from the
probability distribution on the alignments on extant vertices. We use
very basic techniques such as pairwise comparisons between extant
sequences to exploit the correlation between them to reconstruct their
pedigree.

Let $P$ be a pedigree and let $T$ be an undirected $X$-forest. We define
$n(G > T:P) := |\{G \in \mcg_P : T_u(G) \cong T\}|$.

\begin{prop}\label{prop-ngt} Let $P$ be a pedigree with $e(P) = 2e$
  arcs. Let $C \in \Sigma^X$ be any character. Then the probability that
  the $k$-th character in an alignment is $C$ is given by
  \begin{equation}
    \pr\{C_k = C  \cond  P, RM(p,\mu)\} =
    \sum_{T \in \lVert \mcu_P \rVert}
    \frac{n(G > T:P)}{2^e} \pr\{C \cond T, M(\mu)\}.
  \end{equation}
  In particular, it does not depend on $k$. 
\end{prop}

\begin{proof}
  Let $G_1, G_2, \ldots $ be a spanning forest sequence. It is a
  time-homogeneous Markov chain on $\mcg_P$, with transition
  probabilities given by

  \begin{equation*}
    \pr\{G_{j+1} = G^{\prime}\mid G_{j}=G\} 
    = p^{|E(G^\prime)\bigtriangleup E(G)|/2}(1-p)^{|E(G^{\prime})\cap
      E(G)|}.
  \end{equation*}

  It follows that $\pr\{G_k = G\} = 1/2^{e}$ for all $k \in \zz_+$. Let
  $C_1,C_2,\ldots \in \Sigma^{X}$ be a sequence of characters. Then,
  under Model RM,
  \begin{eqnarray}
    &&\pr\{C_k = C \cond P, RM(p,\mu)\} \\[2mm] &=& \sum_{G \in \mcg_P}
    \pr\{C_k = C\mid G_k=G\} \pr\{G_k = G \cond P, RM(p,\mu)\} \nonumber
    \\[2mm] &=& \frac{1}{2^e}\sum_{G \in \mcg_P} \pr\{C_k = C\mid G_k =
    G\} \nonumber \\[2mm] &=& \sum_{T \in \lVert \mcu_P \rVert}
    \frac{n(G > T:P)}{2^e} \pr\{C\mid T, M(\mu)\}.  \nonumber
  \end{eqnarray}
  A similar result holds when Model RM is replaced by Model R.
\end{proof}
The above proposition implies that if two pedigrees have the same number
of undirected $X$-forests of each type and the same number of arcs, then
the character frequencies in alignments alone are not sufficient to
distinguish the two pedigrees. For example, pedigrees in
Figure~\ref{fig-hap} cannot be easily distinguished.

\begin{figure}[ht]
\begin{center}
  \includegraphics[width=118mm]{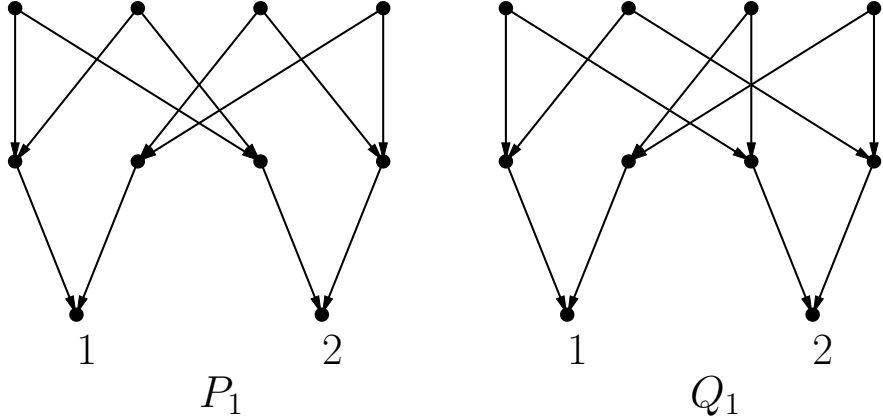}
\end{center}
  \caption[]{Nonisomorphic haploid pedigrees}
  \label{fig-hap}
\end{figure}

But it turns out that, under Model R, most diploid pedigrees of depth 2
are easily distinguished by making pairwise comparisons between extant
sequences and computing the probability that they agree at a site. The
above proposition implies that the probability that two sequences in an
alignment agree at a site $k$ does not depend on $k$ in Models R and
RM. Moreover, such a probability can be computed easily under Model
R. Any given undirected $X$-forest $T$ of a pedigree induces a partition
of its leaves so that the leaves within a component are in the same
part.  Under Model R, the probability that the sequences at two leaves
$i$ and $j$ are in the same state at a site is 1 if they are in the same
component of $T$. Otherwise, the probability is $1/|\Sigma|$. Under
Model RM, the probability depends on $\mu$ if they are in the same
component, and is $1/|\Sigma|$ otherwise.

\begin{prop} \label{prop-2seq} Let $P$ be a diploid pedigree of depth
  2. Let $i,j \in X$. Then the subpedigree of $i$ and $j$ is determined
  by the probability distribution induced on $\Sigma^{\{i,j\}}$ under
  Model R.
\end{prop}

\begin{proof} 
  There are four possible ways in which any two vertices $i$ and $j$ are
  related in a diploid pedigree, which are shown in
  Figure~\ref{fig-grand}. For each of them, we give the probability that
  the sequences $A_i$ and $A_j$ match at any site $k$.  In the
  following, we set $\delta := \frac{1} {16|\Sigma|}$.

  \begin{figure}[ht]
    \begin{center}
      \includegraphics[width=118mm]{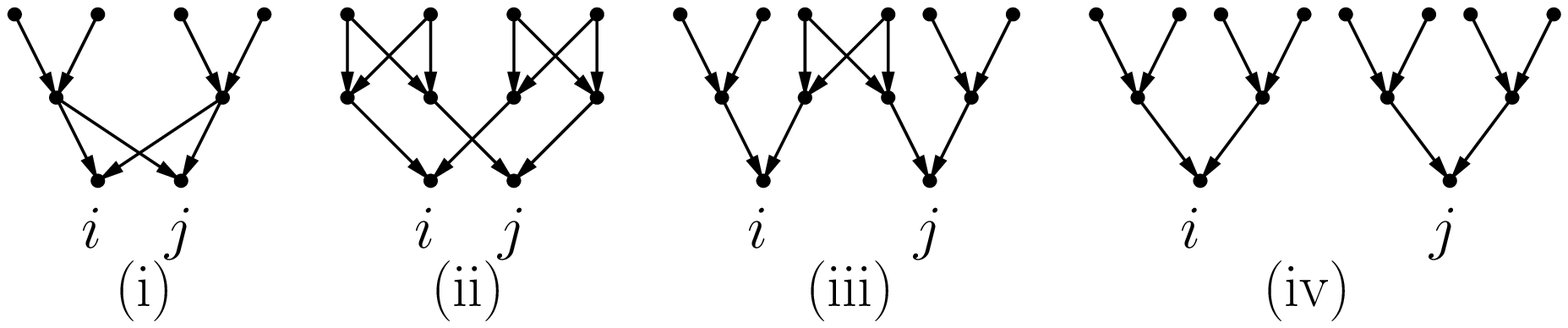}
    \end{center}
    \caption[]{Four ways in which $i$ and $j$ may be related}
    \label{fig-grand}
  \end{figure}

  \noindent {\em Case 1:} If $i$ and $j$ have the same parents, then
  $\pr\{A(i,k) = A(j,k)\} = 8\delta + (1/2)$.

  \noindent {\em Case 2:} If $i$ and $j$ have distinct pairs of parents
  but the same grand parents, then $\pr\{A(i,k) = A(j,k)\} = 12\delta +
  (1/4)$.

  \noindent {\em Case 3:} If $i$ and $j$ have distinct pairs of parents
  but exactly one pair of grand parents in common, then $\pr\{A(i,k) =
  A(j,k)\} = 14\delta + (1/8)$.

  \noindent {\em Case 4:} If $i$ and $j$ have no common parents or grand
  parents, then $\pr\{A(i,k) = A(j,k)\} = 16\delta$.

  Therefore, unless $|\Sigma| = 1$, the above cases are distinguished by
  the marginal joint distribution on $\Sigma^{\{i,j\}}$ under Model R.
\end{proof}

\begin{rem} The assumption that $P$ is a diploid pedigree is
  essential. It implies that $i$ and $j$ are not related as in the
  haploid pedigree $Q_1$ shown in Figure~\ref{fig-hap}. We can verify
  that $\pr\{A(i,k) = A(j,k)\} = 12\delta + (1/4)$ for both $P_1$ and
  $Q_1$, so they are indistinguishable by the above method, which was
  also pointed out as a consequence of Proposition~\ref{prop-ngt}.
\end{rem}

\begin{rem} The above probabilities do not depend on $p$, therefore,
  we have a slightly stronger identifiability statement: If diploid
  pedigrees $P$ and $Q$ of depth 2 and crossover probabilities $p$ and
  $p^\prime$ are such that $(\Sigma^{X} : P,R(p)) = (\Sigma^{X} :
  Q,R(p^\prime))$ then their subpedigrees of order 2 are correspondingly
  isomorphic.
\end{rem}

\begin{prop}\label{prop-n2-1} 
  When $|\Sigma| > 2$, pedigrees of depth 2 in which no two vertices
  have exactly one common parent are identifiable under model R. In
  particular, when $|\Sigma| > 2$, diploid pedigrees of depth 2 are
  identifiable under model R.
\end{prop}

\begin{proof} There are only 4 ways in which any two extant vertices $i$
  and $j$ are related. They are illustrated in
  Figure~\ref{fig-grand}. Each of the possible relationships is
  recognised by Proposition~\ref{prop-2seq}. We denote the 3rd and the
  4th types of relationships by $i \sim j$ and $i \not\sim j$,
  respectively.

  Suppose that no two extant vertices $i, j$ in a pedigree are related
  to each other as $i \sim j$ or $i \not\sim j$. Then the pedigree is
  constructed by adding extant vertices one by one. On each step we add
  one extant vertex and join to previously added extant vertices as in
  Figure~\ref{fig-grand} - (i) or (ii), whichever is
  appropriate. Therefore, we assume that at least two extant vertices
  $i,j$ are related as $i \sim j$ or $i \not\sim j$.

  Let $Z$ be a (nonempty) maximal subset of $X$ such that for any two
  distinct extant vertices $i$ and $j$ in $Z$, either $i \sim j$ or $i
  \not\sim j$. Every other extant vertex $k$ not in $Z$ is related to
  some vertex in $Z$ as in Figure~\ref{fig-grand} - (i) or
  (ii). Therefore, once the subpedigree of $Z$ is constructed, there is
  only one way to extend it to the whole pedigree.

  To construct the subpedigree of $Z$, we first construct an edge
  labelled graph with edge set $Z$ in which edges $i$ and $j$ are
  incident if and only if $i \sim j$. This is a known problem in graph
  theory, namely, the problem of constructing an edge labelled graph
  from its line graph. It was proved \cite{whitney32} that there are
  only 4 pairs $(G_i,H_i)$ of connected nonisomorphic edge labelled
  graphs that have the same line graphs.  Edge labelled graphs that
  cannot be uniquely constructed from their line graphs must contain
  components isomorphic to $G_i$ or $H_i$. We refer to \cite{lovasz93}
  (in particular, Chapter 15, Problem 1) for discussion about
  reconstructing graphs from their line graphs, in particular, for the
  complete list of pairs $(G_i,H_i)$. The first pair is $(K_{1,3},K_3)$
  (with edges of each of them labelled $i,j,k$). Based on the example
  $(K_{1,3},K_3)$, we construct pedigrees shown in Figure~\ref{fig-k3},
  in which all pairs of extant vertices are similarly related.

  \begin{figure}[ht]
    \begin{center}
      \includegraphics[width=118mm]{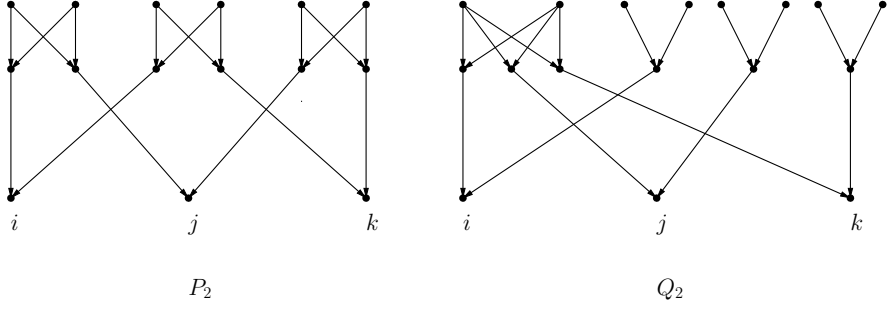}
    \end{center}
    \caption[]{Pedigrees indistinguishable by site pattern
      probabilities}
    \label{fig-k3}
  \end{figure}

  We distinguish the two pedigrees in Figure~\ref{fig-k3} by comparing
  the probabilities $\pr\{A(i,s) = A(j,s) = A(k,s) \cond P_2, R(p)\}$
  and $\pr\{A(i,s) = A(j,s) = A(k,s) \cond Q_2, R(p)\}$ for any site
  $s$.

  In $P_2$, there are 512 spanning forests. Among them there are 192
  spanning forests in which two extant vertices have a common
  grandparent, giving the first term on the RHS below. In the remaining
  320 spanning forests, no two extant vertices have a common
  grandparent, which explains the second term on the RHS
  below. Therefore,
  \[
  \pr\{A(i,s) = A(j,s) = A(k,s) \cond P_2, R(p)\} =
  \frac{192}{512|\Sigma|} + \frac{320}{512|\Sigma|^2}.
  \]
  
  In $Q_2$, there are 512 spanning forests. Among them there are 16
  spanning forests in which $i,j,k$ have a common grandparent (giving
  the first term on the RHS below), 144 spanning forests in which two
  extant vertices have a common grandparent at site $s$ (giving the
  second term), and 352 spanning forests in which $i,j,k$ have distinct
  grandparents at site $s$ (giving the third term).  Therefore,

  \[
  \pr\{A(i,s) = A(j,s) = A(k,s) \cond Q_2, R(p)\} = \frac{16}{512} +
  \frac{144}{512|\Sigma|} + \frac{352}{512|\Sigma|^2}.
  \]
  Whenever $|\Sigma| > 2$, $\pr\{A(i,s) = A(j,s) = A(k,s) \cond Q_2,
  R(p)\} > \pr\{A(i,s) = A(j,s) = A(k,s) \cond P_2, R(p)\}$, therefore,
  $P_2$ and $Q_2$ can be distinguished. The two expressions are equal
  when $|\Sigma| = 2$. Similarly, other pedigrees constructed from
  $(G_i,H_i), i = 2,3,4$ are distinguished when $|\Sigma| > 2$.
  
  Since there are only two types of site patterns for three sequences
  when $|\Sigma| = 2$ (either the three sequences agree at a site or
  exactly two of them agree), the two cases cannot be distinguished by
  considering other site pattern probabilities.
  
  In diploid pedigrees, no two vertices have exactly one common parent,
  therefore, when $|\Sigma| > 2$, they are identifiable under model R.
\end{proof}

We used only site pattern probabilities in the above proofs. But because
of recombinations, consecutive sites in an alignment are not
independent. We use the dependence between sites to eliminate the
restriction $|\Sigma| > 2$ when the crossover probability $p$ is
sufficiently small.

\begin{prop}\label{prop-n2-2}
  When $|\Sigma| = 2$ and $p$ is sufficiently small, haploid pedigrees
  $P_2$ and $Q_2$ (shown in Figure~\ref{fig-k3}) are distinguished under
  model $R(p)$.
\end{prop} 

\begin{proof} We compute the probability that, in an alignment $A$,
  there are long runs of sites at which all the three sequences $A_i$,
  $A_j$ and $A_k$ are equal. In particular, we compute bounds on
  $\pr\{A(i,m) = A(j,m) = A(k,m) \forall m \in [1,t+1]\}$ on the two
  pedigrees.

  For $P_2$, for any fixed $m$, if any two of the three sites $A(i,m),
  A(j,m),A(k,m)$ are inherited from the same grand parent, then
  $\pr\{A(i,m) = A(j,m) = A(k,m)\} = 1/2$, and if all of them are
  inherited from distinct grandparents, then $\pr\{A(i,m) = A(j,m) =
  A(k,m)\}$ is 1/4. Therefore,
  \[
  \pr\{A(i,m) = A(j,m) = A(k,m) \forall m \in [1,t+1] \cond P_2,R(p) \}
  \leq (1/2)^{t+1}.
  \]

  For $Q_2$, the probability that the first site of all sequences is
  inherited from a common grand parent is 1/32. The probability that at
  each successive site the three sequences have a common grand parent is
  $(1-p)^6+p^3(1-p)^3$. Therefore,
  \begin{eqnarray*}
    && \pr\{A(i,m) = A(j,m) = A(k,m) \forall m \in [1,t+1] \cond
    Q_2,R(p) \} \nonumber \\[2mm] &\geq&
    \frac{\left((1-p)^6+p^3(1-p)^3\right)^t}{32}.
  \end{eqnarray*}
  When $p$ is sufficiently small and $t$ is sufficiently large, the
  above probability for $Q_2$ is more than that for $P_2$.
\end{proof}

We do not analyse other examples of pedigrees based on graphs that are
not reconstructible from their line graphs (graphs $G_i, H_i, i = 2,3,4$
mentioned in Proposition~\ref{prop-n2-1}), but they may be analysed
similarly.  

\section{Reconstructing pedigrees under Model RM}
\label{sec-model-rm}

In this section we develop ideas from Section~\ref{sec-model-r}
(especially Proposition~\ref{prop-n2-2}) in much more
generality. Earlier we observed that parts of alignments that are free
of recombination may be analysed with phylogenetic methods. Since we do
not know where the recombinations have occurred in an alignment, we
choose long segments of carefully chosen alignments (or sets of
alignments) and show that they have higher probability of having evolved
on one particular $X$-forest (or a sequence of $X$-forests) than any
other $X$-forest (or a sequence of $X$-forests). For example, the method
of Proposition~\ref{prop-n2-2} works because $Q_2$ contains a tree in
which $i,j,k$ have a common ancestor and $P_2$ does not contain such a
tree. Therefore, a sufficiently long sequence of characters in which
sequences $A(i),A(j)$ and $A(k)$ are in the same state will more likely
have evolved on a pedigree such as $Q_2$ than on a pedigree that does
not contain such a tree. This argument may be generalised to count the
number of $X$-forests of each type from the distribution of extant
sequences. Such a generalisation requires identifiability results for
phylogenetic trees, which we state in a form suitable for our
application.

\subsection{Identifiability and consistency results for $X$-forests}
\label{sec-consistency}

In this section, we state known results on identifiability and
statistical consistency of maximum likelihood reconstruction of
phylogenetic trees.  We need to adapt them slightly since the
$X$-forests in a pedigree differ from phylogenetic trees in three
respects - they may have vertices of degree 2, they may be unresolved
(i.e., they may have vertices of degree more than 3), and they may be
disconnected (two extant vertices may not have a common ancestor in a
given directed $X$-forest in a pedigree). We address them in the
following identifiability result, which was originally proved for
phylogenetic trees in \cite{pt86} in the $|\Sigma| = 2$
case. Identifiability and statistical consistency of maximum likelihood
reconstruction of phylogenetic trees were independently proved in full
generality for all $|\Sigma| \geq 2$ in \cite{steel94,chang96}.

\begin{thm}
  \label{thm-identifiability}
  For all $\mu \in (0,1/|\Sigma|)$ and any two undirected $X$-forests
  $T_1$ and $T_2$ with bounded number of edges, if $\pr\{C \cond T_1,
  M(\mu) \} = \pr\{C \cond T_2,M(\mu) \}$ for all characters $C \in
  \Sigma^X$ then $T_1 \cong T_2$.
\end{thm}

\begin{proof}
  The result follows from the analogous results in
  \cite{steel94,chang96} for phylogenetic $X$-trees, but we have to
  clarify three issues: unlike the phylogenetic $X$-trees, the
  $X$-forests as defined in this paper may be disconnected, unresolved,
  and may have vertices of degree 2.

  {\em Connectivity:} Any two extant vertices $x_i$ and $x_j$ are in
  different components of $T_1$ and $T_2$ if and only if $\pr\{C(i) = a
  \mid C(j) = b \} = 1/|\Sigma|$ for all $a,b \in \Sigma$, where $C(i)$ is
  the state at the extant vertex $i$ in a character $C$. If $i$ and $j$
  are in the same component, then $\pr\{C(i) = a \mid C(j) = b \}$ cannot
  be arbitrarily close to $1/|\Sigma|$ if the number of edges (and hence
  the distance between $i$ and $j$) is bounded. Therefore, we can
  consider the identifiability question for each component separately.

  {\em Unresolved $X$-forests:} Given an unresolved phylogenetic tree,
  there are resolved phylogenetic trees with site pattern probabilities
  arbitrarily close to the site pattern probabilities for the unresolved
  tree. Therefore, even though unresolved phylogenetic $X$-trees are
  identifiable, statistical consistency of maximum likelihood methods
  requires that the substitution probabilities on the edges of a
  phylogenetic tree are bounded below by a positive real number. In our
  model, Corollary~\ref{cor-consistency} below is possible because the
  substitution probability $\mu$ is fixed on each edge.

  {\em Vertices of degree 2:} Let $u$ and $v$ be any two vertices of an
  $X$-forest. Suppose that $u$ and $v$ are of degree 1 or more than
  2. Suppose all internal vertices on the path between $u$ and $v$ have
  degree 2. Since $\mu$ is fixed for all arcs, the distance between $u$
  and $v$ on a tree is determined by the substitution probability on the
  $uv$ path. The substitution probability on the $uv$ path is determined
  by the distribution on the space of characters.
\end{proof}

\begin{rem} In the above result, if $T_1$ and $T_2$ were directed
  $X$-forests, then we would be able to conclude that $T_u(T_1) \cong
  T_u(T_2)$.
\end{rem}

Let $N := |\Sigma|^{|X|}$. Suppose that $\Sigma^X := \{C_i, i \in
N\}$. We associate with each undirected $X$-forest $T$ a vector
$\mathbf{p}(T, \mu) := (p_1,p_2, \ldots, p_N)$ in $\rr_{+}^N$, where
$p_i := \pr\{C_i \cond T, M(\mu) \}$. Therefore, the condition $\pr\{C
\cond T_1, M(\mu) \} = \pr\{C \cond T_2, M(\mu) \}$ for all characters
$C \in \Sigma^X$ may be equivalently written as $\mathbf{p}(T_1, \mu) =
\mathbf{p}(T_2,\mu)$.

\nomenclature{$\mathbf{p}(T, \mu):= (p_1,p_2, \ldots,p_N)$}{defined as
  $p_i := \pr\{C_i \cond T,M(\mu) \}$}

Given $r \in \rr_{+}$ and a point $\mathbf{s} \in \rr^N$, let the open
ball of radius $r$ centred at $\mathbf{s}$ be denoted by
$\rho(\mathbf{s}, r)$. \nomenclature{$\rho(\mathbf{s}, r)$}{a ball of
  radius $r$ centred at $\mathbf{s}$} Here the radius may be taken to be
in the 1-norm (i.e., the distance between points $\mathbf{x} :=
(x_1,x_2, \ldots, x_N)$ and $\mathbf{y} := (y_1,y_2, \ldots, y_N)$ is
defined by $d(\mathbf{x},\mathbf{y}) = \sum_i^N |x_i - y_i|$).
\nomenclature{$d(\mathbf{x},\mathbf{y})$}{1-norm distance between
  $\mathbf{x}$ and $\mathbf{y}$} Let $A$ be an alignment on $X$. We
define a vector $\mathbf{f}(A) := (f_1, f_2, \ldots, f_N)$,
\nomenclature{$\mathbf{f}(A) := (f_1, f_2, \ldots, f_N)$}{vector of
  fractional site pattern frequencies in an alignment $A$} where $f_i$
are the {\em fractional site pattern frequencies}, i.e., $f_i$ is the
fraction of columns of $A$ of type $C_i$. Then the above identifiability
result implies a statistical consistency result for maximum
likelihood. It informally says that as the length of a random alignment
$A$ goes to infinity, we expect $\mathbf{f}(A)$ to be arbitrarily close
to $\mathbf{p}(T, \mu)$ with probability approaching 1 if $T$ is the
true $X$-forest, and that the probability that $\mathbf{f}(A)$ is
arbitrarily close to $\mathbf{p}(T, \mu)$ approaches 0 if $T$ is not the
true $X$-forest.

\begin{cor}
  \label{cor-consistency}
  For all $r_0 \in \rr_+$, $\epsilon \in (0,1)$ and $\mu \in
  (0,1/|\Sigma|)$, there exists $L := L(r_0, \epsilon, \mu) \in \nn$
  such that for any two undirected $X$-forests $T$ and $T^{\prime}$ such
  that $T \not \cong T^\prime$, and an alignment $A \in \Sigma^{XL}$,
  \[
  \pr\{\mathbf{f}(A) \in \rho(\mathbf{p}(T,\mu), r_0)  \cond  T, M(\mu)
  \} > 1 - \epsilon
  \]
  and
  \[
  \pr\{\mathbf{f}(A) \in \rho(\mathbf{p}(T,\mu), r_0)  \cond 
  T^{\prime}, M(\mu) \} < \epsilon.
  \]
\end{cor}

We give bounds on the above probabilities in terms of $L$, which we
prove using Bernstein's inequality.

\begin{lemma}[Bernstein's inequality]\label{lem-bernstein1}
  Let $X, X_1, X_2, \ldots$ be i.i.d. Bernoulli random variables with
  $\pr\{X = 1\} = p$. Then for all $r \geq 0$ and $n \in \zz_+$,
  \[
  \pr\left\{\frac{\sum_{i=1}^n X_i}{n} - p \geq r \right\} \leq
  \exp\left\{\frac{-nr^2}{2p(1-p) + 2r/3}\right\},
  \]
  and (equivalently)
  \[
  \pr\left\{\frac{\sum_{i=1}^n X_i}{n} - p \leq -r \right\} \leq
  \exp\left\{\frac{-nr^2}{2p(1-p) + 2r/3}\right\}.
  \]
\end{lemma}

\begin{lemma}\label{lem-eps-l1}
  Let $r_0 \in \rr_+$ and $\mu \in (0,1/|\Sigma|)$. Let $A \in
  \Sigma^{XL}$. Let $T$ be an undirected $X$-forest. Then
  \[
  \pr\{\mathbf{f}(A) \not \in \rho(\mathbf{p}(T,\mu), r_0)  \cond  T,
  M(\mu) \} \leq 2 |\Sigma|^{|X|} \exp \left \{\frac{-Lr_0^2}
    {\frac{|\Sigma|^{2|X|}}{2} + \frac{2r_0 |\Sigma|^{|X|}}{3}}
  \right \}.
  \]
\end{lemma}
\begin{proof} There are $|\Sigma|^{|X|}$ distinct characters, with
  probabilities $p_i := \pr\{C_i \cond T,M(\mu) \}$, $i = 1\;\text{to}\;
  |\Sigma|^{|X|}$. If $\mathbf{f}(A) \not \in \rho(\mathbf{p}(T,\mu),
  r_0)$, then $|f_i - p_i| \geq r_0/|\Sigma|^{|X|}$ for some
  $i$. Therefore, we apply Bernstein's inequality to each distinct
  character, and write the probability that $|f_i - p_i| \geq
  r_0/|\Sigma|^{|X|}$. We then apply the union bound to get the result.
\end{proof}

If we set
\begin{equation}\label{eq-eps0}
  r_0 = \frac{\min\{d(\mathbf{p}(T_i, \mu), \mathbf{p}(T_j, \mu)):
    T_i,T_j \in \mcu, T_i \not \cong T_j\}}{2},
\end{equation}
it will ensure that the $X$-forests are {\em separated} by open balls of
radius $r_0$ in the space of site pattern probability vectors,
i.e., the open balls $\rho(\mathbf{p}(T_i,\mu), r_0)$ and
$\rho(\mathbf{p}(T_j,\mu), r_0)$ are non-intersecting whenever
$T_i$ and $T_j$ are non-isomorphic. We will use this value of
$r_0$ unless specified otherwise.

Now for an undirected $X$-forest $T_i$, we define
\begin{equation} \label{eq-mcai} \mca_{i} := \mca(T_i,r_0,L) :=
  \{A \in \Sigma^{XL} : \mathbf{f}(A) \in
  \rho(\mathbf{p}(T_i,\mu),r_0)\},
\end{equation}
\nomenclature{$\mca(T_i,r_0,L)$}{the set of alignments of length $L$
  whose fractional site pattern frequencies are within a radius $r_0$
  from $\mathbf{p}(T_i)$} and
\begin{equation}\label{eq-eti}
\epsilon_i := \epsilon(T_i) := 1 - \pr\{\mca_{i} \cond T_i, M(\mu) \}.
\end{equation}
By selecting a sufficiently large value of $L$ we can make $\epsilon_i$
arbitrarily small as in Lemma~\ref{lem-eps-l1}. Moreover, if $ T_i \not
\cong T_j$, then
\begin{equation}\label{eq-p-mcai}
  \epsilon_{ij} := \pr\{\mca_{i} \cond T_j, M(\mu) \} \leq 1 - \pr\{\mca_j
  \cond T_j, M(\mu) \})= \epsilon_j,
\end{equation}
hence $\pr\{\mca_{i} \cond T_j, M(\mu) \}$ can be made arbitrarily small
as per Lemma~\ref{lem-eps-l1}. We set $\epsilon_{max} = \max_i
(\epsilon(T_i))$, which depends on $L$ and $r_0$.

In Theorem~\ref{thm-paths} (particularly in the proof of
inequality~\ref{eq-term1a}) we require a concentration inequality
similar to the inequality in Lemma~\ref{lem-eps-l1} for the situation in
which $L$ sites of an alignment have evolved on an $X$-forest $T_i$ and
$cL$ sites (for a small $c \in (0,1)$) have evolved on another
$X$-forest $T_j$. (We will keep the notation simple by assuming that
$cL$ is an integer.) Therefore, we give the following variants of
Lemmas~\ref{lem-bernstein1} and ~\ref{lem-eps-l1}.

\begin{lemma}\label{lem-bernstein2} Let $X, X_1, X_2, \ldots $ be
  i.i.d. Bernoulli random variables with $\pr\{X=1\} = p$. Let $Y_1,
  Y_2, \ldots$ be Bernoulli random variables. Let $S_n := \sum_{i=1}^n
  X_i + \sum_{i=1}^{cn} Y_i$, where $c$ is a positive constant. Let $r
  \geq 0$ and $m := \max \{0,p-r\}$ and $M := \min \{1,p+r\}$. If
  $r^\prime := p-m(1+c) \geq 0$ and $r^{\prime\prime} := M(1+c)-(p+c)
  \geq 0$, then
  \[
  \pr\left\{\left\vert \frac{S_n}{n(1+c)} - p \right\vert \geq r \right\}
  \leq \exp\left\{\frac{-n(r^\prime)^2}{2p(1-p) + 2r^\prime/3 } \right\}
  + \exp\left\{\frac{-n(r^{\prime\prime})^2}{2p(1-p) + 2
      r^{\prime\prime}/3 } \right\}.
  \]
\end{lemma}
\begin{proof}
  Since $\sum_{i=1}^n X_i \leq S_n \leq cn+\sum_{i=1}^n X_i$, we have
  \begin{eqnarray*}
    && \left\vert \frac{S_n}{n(1+c)} - p \right\vert \geq r \\[1mm]
    & \implies & \left(\frac{S_n}{n(1+c)} - p \leq -r\right)
    \;\text{or}\; \left(\frac{S_n}{n(1+c)} - p \geq r\right)
    \\[1mm]
    & \implies & \left(\frac{\sum_{i=1}^n X_i}{n(1+c)} - p \leq -r\right)
    \;\text{or}\; \left(\frac{cn + \sum_{i=1}^n X_i}{n(1+c)} - p \geq
      r\right) \\[1mm]
    & \implies & \left(\frac{\sum_{i=1}^n X_i}{n(1+c)} \leq p-r \leq m \right)
    \;\text{or}\; \left(\frac{cn + \sum_{i=1}^n X_i}{n(1+c)} \geq p+r \geq M
    \right) \\[1mm]
    & \implies & \left(\frac{\sum_{i=1}^n X_i}{n(1+c)}-\frac{p}{1+c} \leq
      m-\frac{p}{1+c}\right)
    \text{or}\; \left(\frac{cn + \sum_{i=1}^n X_i}{n(1+c)} -
      \frac{p+c}{1+c} \geq  M-\frac{p+c}{1+c} \right) \\[1mm]
    & \implies & \left(\frac{\sum_{i=1}^n X_i}{n}-p \leq -r^\prime \right)
    \;\text{or}\; 
    \left(\frac{\sum_{i=1}^n X_i}{n}-p \geq r^{\prime\prime}\right).
  \end{eqnarray*}
  Now we apply Bernstein's inequality (Lemma~\ref{lem-bernstein2}) to
  each term and obtain the desired bound.
\end{proof}

Let $A$ be an alignment of length $L(1+c)$. Suppose that $L$ characters
of $A$ evolved on an undirected $X$-forest $T$ and the remaining
characters evolved on undirected $X$-forests $T_1, T_2, \ldots,
T_{cL}$. The following lemma states that if $c$ is sufficiently small,
then $\mathbf{f}(A)$ is concentrated near $\mathbf{p}(T,\mu)$ for large
$L$.  Moreover, as in Lemma~\ref{lem-eps-l1}, if we require
$\mathbf{f}(A)$ to be sufficiently near $\mathbf{p}(T,\mu)$ with
probability at least $1 - \epsilon_{max}$, then the length of the
alignment $L(1+c)$ must be $\Omega(\log (1/\epsilon_{max}))$. In the
following lemma, we do not specify the constants $c, c_i, r_i^\prime$
and $r_i^{\prime\prime}$ precisely, but they can be chosen depending on
$r_0$.

\begin{lemma}\label{lem-eps-l2}
  Let $r_0 \in \rr_+$ and $\mu \in (0,1/|\Sigma|)$. Let $A \in
  \Sigma^{XL(1+c)}$ for a suitably chosen positive constant $c$. Let $T,
  T_1, T_2, \ldots, T_{cL}$ be undirected $X$-forest. Then
  \begin{eqnarray*}
    && \pr\{\mathbf{f}(A) \not \in \rho(\mathbf{p}(T,\mu), r_0) \cond
    T^L,T_1,T_2,\ldots,T_{cL}, M(\mu) \} \\[1mm]
    &\leq&
    \sum_i^{|\Sigma|^{|X|}}\left(\exp \left\{\frac{-L(r_i^\prime)^2}{2p_i(1-p_i)
          + 2r_i^\prime/3}\right\} + \exp
      \left\{\frac{-L(r_i^{\prime\prime})^2} 
        {2p_i(1-p_i) + 2r_i^{\prime\prime}/3}\right\}\right),
  \end{eqnarray*}
  where $r_i^\prime$ and $r_i^{\prime\prime}$ are positive constants as
  in Lemma~\ref{lem-bernstein2}.
\end{lemma}
\begin{proof} We apply Lemma~\ref{lem-bernstein2} for each component of
  $\mathbf{f}(A)$ and use the union bound as in the proof of
  Lemma~\ref{lem-eps-l1}. For each component, we use $r :=
  r_0/|\Sigma|^{|X|}$ as before. Constants $r_i^\prime$ and
  $r_i^{\prime\prime}$ (and $c_i$, which are implicit) depend on $r_0$
  and the probabilities $p_i := \pr\{C_i \cond T,M(\mu) \}$, $i =
  1\;\text{to}\; |\Sigma|^{|X|}$. The constant $c$ may be taken to be
  the smallest among $c_i, i = 1\;\text{to}\; |\Sigma|^{|X|}$.
\end{proof}

\subsection{Identifying HMMs: a sketch of the ideas used to prove
  Theorem~\ref{thm-paths}}\label{sec-sketch}

A hidden Markov model (HMM) is defined by two sequences $\{X_n\}_{n \geq
  1}$ and $\{Y_n\}_{n \geq 1}$ of random variables. The sequence
$\{X_n\}_{n \geq 1}$ takes values in $[r]$ and is a stationary Markov
chain with transition matrix $A$ and initial distribution $\pi(i), i =
1$ to $r$, which is also the stationary distribution of the Markov
chain. The random variables $\{Y_n\}_{n \geq 1}$ take values in $[k]$,
and are independent and identically distributed conditional on
$\{X_n\}_{n \geq 1}$. The distribution of $Y_n$ depends only on
$X_n$. Let $B$ be the $r \times k$ matrix of conditional probabilities
$\pr\{Y_n = j \mid X_n = i \}$, where $i \in [r]$ and $j \in [k]$. The
sequence $\{Y_n\}_{n \geq 1}$ are the {\em observations}. Identifiable
hidden Markov models were characterised in \cite{petrie69}, where a
precise description of conditions on $A$ and $B$ for which the
probability distribution on observed sequences determines $A$ and $B$
(up to re-labelling of hidden states) was given. Here {\em
  identifiability up to a re-labelling of hidden states} means the
following: If $S$ is an $r\times r$ permutation matrix, then the HMM
with parameters $(A,B,\pi)$ (where $\pi$ is treated as a column vector
of length $r$) is equivalent to (induces the same distribution on the
space sequences of observed states as) the HMM with parameters
$(S^{-1}AS,S^{-1}B,S^{-1}\pi)$. Therefore, identifiability only means
computing the matrices and the initial distribution up to
equivalence. We denote the class of models equivalent to $(A,B,\pi)$ by
$\lVert (A,B,\pi)\lVert$.

Earlier we noted that Models R and RM for sequences evolving on a
pedigree $P(X)$ define a hidden Markov model with the spanning forests
in $P$ as hidden states and characters from $\Sigma^X$ as observed
states. We call it $HMM(P,p,\mu)$ and denote its matrices by $A(P,p)$
and $B(P,\mu)$. The initial distribution on hidden states is uniform:
each spanning forest has the probability $1/2^e$ if the pedigree has
$2e$ arcs. We informally look at some of the issues about its
identifiability.

The transition matrix $A(P,p)$ is defined by transition probabilities
given in Equation~(\ref{eq-gi-gj}). Therefore, $A(P,p)$ will be
identical (up to a permutation of rows and columns) for all pedigrees
with the same number of arcs, for a fixed $p$. But the set of spanning
forests (hidden states) is unknown.

We now describe at a high level how we compute the rows of
$B(P,\mu)$. Suppose the pedigree contains an undirected $X$-forest
$T_i$. There are $n(G > T_i:P)$ spanning forests $G$ of $P$ that contain
$T_i$ as the unique undirected $X$-forest, and corresponding to each of
them we have a row of $B(P,\mu)$ that is equal to
$\mathbf{p}(T_i,\mu)$. Now consider a set $\mca_{i} := \mca(T_i,r_0,L)$
of sufficiently long alignments as defined in Equation~\ref{eq-mcai}. We
compute the probability of $\mca_{i}$ (i.e., the probability that a
random alignment is in $\mca_{i}$). Suppose that $P_0$ is the
probability that there are no recombination events. Thus $P_0$
approaches 1 as $p$ approaches 0. Then one of the terms in the
expression for the probability of $\mca_{i}$ will be $n(G > T_i:P) P_0
(1-\epsilon_i)/2^e$, where $n(G > T_i:P)/2^e$ is the probability that
the first site evolved on a spanning forest $G$ that contained $T_i$ as
the undirected $X$-forest. There will be terms for contributions from
other undirected $X$-forests $T_j\not \cong T_i$, but they will be much
smaller than the above term because they will contain factors
$\epsilon_j$ (as in Equation~\ref{eq-p-mcai}). There will also be terms
to account for recombinations among the first $L$ sites, but they will
be small as well since they will contain $p$ as a factor (in contrast to
$P_0$, which is a power of $(1-p)$). So let us say $\pr\{\mca_{i} \cond
P,RM(p,\mu))$ is $n(G > T_i:P) P_0 (1-\epsilon_i)/2^e +
(\mathtt{terms\;\; of\;\; smaller\;\; order})$. Therefore if $p$ is
sufficiently small and $L$ is sufficiently large, then $\pr\{\mca_{i}
\cond P,RM(p,\mu))$ will be roughly equal to the dominating term $n(G >
T_i:P) P_0 (1-\epsilon_i)/2^e$, which will uniquely determine $n(G >
T_i:P)$. In other words, if $Q$ is another pedigree such that $n(G >
T_i:P) \neq n(G > T_i:Q)$, then $\pr\{\mca_{i} \cond P,RM(p,\mu))$ and
$\pr\{\mca_{i} \cond Q,RM(p,\mu))$ will differ roughly by a multiple of
$P_0 (1-\epsilon_i)/2^e$.

In the proof of Proposition~\ref{prop-n2-2}, we used a similar idea: the
pedigree $Q_2$ contains a certain subtree $T$ in which $i,j,k$ have a
common ancestor, while the pedigree $P_2$ does not such a subtree. As a
result, the alignments that are close to $\mathbf{p}(T,\mu)$ are more
likely to have evolved on $Q_2$ than on $P_2$.

Suppose now that $B(P,\mu)$ is identified and each of its rows is
labelled by the corresponding unlabelled undirected $X$-forest. That is,
the matrices $B(P,\mu)$ that appear among the triples in the equivalence
class $\lVert (A,B,\pi)\lVert$ of HMMs are constructed. As pointed out
above, the matrix $A(P,p)$ and the initial distribution are also known
up to relabelling of hidden states. But the equivalence class $\lVert
(A,B,\pi)\lVert$ is not known unless we are able to label the rows and
the columns of $A(P,p)$ by unlabelled undirected $X$-forests in a manner
consistent with the labelling of rows of $B(P,\mu)$. In other words, for
full identifiability of $HMM(P,p,\mu)$, we would like to construct an
automaton with transition probabilities given by $A(P,p)$ and with its
states labelled by unlabelled undirected $X$-forests. Identifying the
pedigree from the {\em labelled automaton} will then be a purely
combinatorial problem.

In this paper we do not succeed in constructing matrix $A(P,p)$ with
rows and columns labelled by undirected $X$-forests, but we are able to
count certain types of {\em walks} (to be described next) on the
automaton with vertices labelled by undirected $X$-forests.  Suppose
that $T_1, T_2, \ldots, T_m$ is a sequence of undirected $X$-forests
such that no two consecutive ones are isomorphic. Analogous to $n(G >
T:P)$, we define $n(\mathbf{G} > \mathbf{T}:P) $ as the number of
sequences $G_1, G_2, \ldots, G_m$ of spanning forests in $P$ such that
$G_i > T_i$, where consecutive $G_i$ are separated by just one
recombination.  (A single recombination at a site is more likely than
multiple recombinations. Moreover, if there is a recombination at a site
$i$, but the two spanning forests $G_i$ and $G_{i+1}$ contain isomorphic
undirected $X$-forests, then such a recombination has no effect on the
emitted characters. These are the reasons why we consider the sequences
$T_i$ and $G_i$ as above.) We then consider a set $\mca$ of alignments
of length $mL$ (for a suitably large $L$) obtained by concatenating
alignments from $\mca_{i}$ for $i = 1$ to $m$. We compute the
probability of $\mca$ (as we described for $n(G>T)$ above), and show
that the dominating term is proportional to $n(\mathbf{G} >
\mathbf{T}:P) $, and other terms are of smaller order of magnitude for
small $p$. This allows us to compute $n(\mathbf{G} > \mathbf{T}:P)$.

A more visual description of the walks may be given as follows. Suppose
the pedigree has $2e$ arcs. We define a graph on the vertex set
consisting of the spanning forests of the pedigree, with two spanning
forests $G_i$ and $G_j$ being adjacent if there is exactly one
recombination separating them (i.e., $|E(G_{i})\bigtriangleup E(G_{j})|
= 2$). The graph is a hypercube. The hidden Markov chain on the set of
spanning forests jumps on the vertices of the cube. If there is at most
one recombination at any site (which is more likely than more than 1
recombinations at a site), then we have a walk on the edges of the
cube. We label each vertex $G_i$ of the cube by the undirected
$X$-forests $T_u(G_i)$. Our interest is to construct this object for a
more complete understanding of the HMM. But problem is made difficult by
the fact that the emission probabilities associated with $G_i$ and $G_j$
are identical if $T_u(G_i) \cong T_u(G_j)$. Therefore, we construct a
weaker object, namely the number of walks of each length on the cube
such that consecutive vertices have distinct labels.

\subsection{The main results}\label{sec-main}
\begin{defn}
Let $P$ be a pedigree. For $\mathbf{T} := (T_1,T_2,\ldots, T_m) \in
\mcu_P^m$, we define
\begin{eqnarray*}
  && n(\mathbf{G} > \mathbf{T}:P) \\[2mm]
  &:=& n(G_1 > T_1, G_2 > T_2, \ldots, G_m > T_m:P) \\[2mm]
  &:=& |\{\mathbf{G} \in \mcg_P^m : T_u(G_i) \cong T_i \forall
  i \in [m] \wedge |E(G_{i+1})\bigtriangleup E(G_{i})| = 2 \; \forall\;
  i \in [m-1]\}|,
\end{eqnarray*}
where the second condition in the last line says that there is exactly
one recombination event between consecutive $G_i$.
\end{defn} 

\nomenclature{$n(\mathbf{G} > \mathbf{T}:P)$}{number of sequences
  $\mathbf{G}$ of spanning forests in $P$ for which $T_u(G_i) \cong T_i$
  for all $G_i$ in $\mathbf{G}$ and consecutive $G_i$ are separated by
  exactly 1 recombination}

In the rest of this section, we show how invariants $n(\mathbf{G} >
\mathbf{T}:P)$ may be computed from the probability distribution on the
space of alignments under Model RM. In the end, we demonstrate an
application to pedigrees $P_1$ and $Q_1$ shown in Figure~\ref{fig-hap}.

\begin{lemma}\label{lem-pg}
  Let $P$ be a pedigree with $e(P) = 2e$ arcs. Let $\mathbf{G} := (G_1,
  G_2,\ldots, G_m)$ be a sequence of spanning forests in $P$. Then under
  model $R(p)$, the probability that $\mathbf{G}$ is a sequence of
  site-specific spanning forests is given by
  \[
  \pr\{\mathbf{G}  \cond  P, R(p) \} =
  \frac{(1-p)^{s(\mathbf{G})}p^{r(\mathbf{G})}}{2^{e}},
  \]
  where $r(\mathbf{G})$ and $s(\mathbf{G})$ are as in
  Definition~\ref{def-rg-sg}.
\end{lemma}

\begin{proof}
  We have a factor $p$ for each recombination event and $(1-p)$ whenever
  there is no recombination (i.e., an arc is contained in two
  consecutive spanning forests in the sequence). The probability that
  the first spanning forest is $G_1$ is $1/2^e$.
\end{proof}

\subsection*{Notation} 

Let $A \in \Sigma^{XL}$ be an alignment of length $L$. For an interval
$[l_1,l_2] \subseteq [L]$, we write $A[l_1,l_2]$ for the part of the
alignment between columns $l_1$ and $l_2$ (inclusive of columns $l_1$
and $l_2$). For a sequence of alignments $A_i \in \Sigma^{Xl_i}, i \in
[m]$, let $A:=A_1:A_2:\ldots:A_m$ denote the alignment obtained by
concatenating alignments $A_i, i = 1,2, \ldots, m$ (in that order). Let
$\mca_{i} \subseteq \Sigma^{Xl_i}, i \in [m]$ be sets of alignments. We
define
\[
\mca := \mca_{1}:\mca_{2}:\ldots:\mca_m := \{A_1:A_2:\ldots:A_m \mid
A_i \in \mca_{i}, i \in [m]\}.
\]

\nomenclature{$A:=A_1:A_2:\ldots:A_m$}{an alignment obtained by
  concatenating alignments $A_i, i = 1,2,\ldots,m$} \nomenclature{$\mca :=
  \mca_{1}:\mca_{2}:\ldots:\mca_m$}{a set of alignments obtained by
  concatenating alignments from sets $\mca_{i},  i = 1,2,\ldots,m $}

\begin{lemma}\label{lem-pa}
  Let $P$ be a pedigree with $e(P) = 2e$ arcs. The probability of an
  alignment $A \in \Sigma^{XL}$ on $P$ is given by
  \begin{eqnarray}\label{eq-pa}
    &&\pr\{A \cond  P, RM(p,\mu)\}   \nonumber \\[1mm]
    &=& \sum_{k=1}^{L}\; \sum_{\mathbf{G} \in \mcg_P^k}
    \frac{(1-p)^{e(L-k)+s(\mathbf{G})}p^{r(\mathbf{G})}}{2^{e}} 
    \sum_{\substack{\mathbf{l} \in \zz_+^k : \\
        \sum\limits_{i=1}^kl_i = L}}\;
    \prod_{i=1}^k\pr\{A[L_{i-1}+1,L_i] \cond T_u(G_i),M(\mu) \},   \nonumber \\
    &&
  \end{eqnarray}
  where $L_0 := 0$, $L_i := L_{i-1} + l_i$ and $\mathbf{l} := (l_1, l_2,
  \ldots, l_k) \in \zz_+^k$, and the second summation is over
  $\mathbf{G}$ such that consecutive spanning forests $G_{i}$ and
  $G_{i+1}$ are unequal for $i \in [k-1]$.
\end{lemma}

\begin{proof} The probability of an alignment of length $L$ is obtained
  by summing its probability over all spanning subgraph sequences of
  length $L$. Suppose that the recombinations in a spanning subgraph
  sequence occur only at sites $L_i:= L_{i-1} + l_i $, for $i =
  1,2,\ldots, k-1$, where $L_0 = 0$. We write the spanning forest
  sequence of length $L$ as $(G_i^{l_i}, i = 1,2,\ldots, k)$. Then the
  probability of the alignment is written as a product of probabilities
  of its segments that have evolved on spanning forests $G_i$ (i.e.,
  effectively on $T_u(G_i)$). This probability is summed over
  $\mathbf{l} \in \zz_+^k$ (with the constraint $\sum_i l_i = L$), $k
  \in [L]$ and $\mathbf{G} \in \mcg_P^k$. For a fixed $\mathbf{G}\in
  \mcg_P^k$, the probability of $(G_i^{l_i}, i = 1,2,\ldots, k)$ is
  given by Lemma~\ref{lem-pg}.
\end{proof}

For $\mca \subseteq \Sigma^{XL}$, we will compute $\pr\{\mca  \cond  P,
RM(p,\mu)\}$ by summing (\ref{eq-pa}) over $A \in \mca$. In such a
calculation, we will sometimes use the following upper bound in
evaluating the last summation in (\ref{eq-pa}) for fixed values of $k$,
and fixed $\mathbf{l}$ and $\mathbf{G}$.

\begin{lemma} \label{lem-pA-ub2} Let $\mca \subseteq \Sigma^{XL}$. Let
  $0 = L_0 < L_1 < \ldots < L_k = L$. Let $\mca_{i} :=
  \{A[L_{i-1}+1,L_i] : A \in \mca\}, i \in [k]$. Then for any fixed
  $\mathbf{G} \in \mcg_P^L$
  \begin{equation*}\label{eq-ap-approx}
    \pr\{\mca \cond \mathbf{G}, M(\mu) \} \leq \pr\{\mca_{1}:\mca_{2}:
    \ldots :\mca_{k} \cond \mathbf{G}, M(\mu) \}.
  \end{equation*}
We have equality if $\mca = \mca_{1}:\mca_{2}: \ldots :\mca_{k}$.
\end{lemma}

\begin{proof} The claim follows from the observation that $\mca
  \subseteq \mca_{1}:\mca_{2}: \ldots :\mca_{k}$.
\end{proof}

The following lemma is used in the proof of Equation~(\ref{eq-3b}), and
the reader may skip it until then.

\begin{lemma}\label{lem-seq} Let $\Omega$ be a finite set.
 Let $\mathbf{T}$ and $\mathbf{S}$ be two sequences in $\Omega$, each of
 length $l$, defined by
  \[
  \mathbf{T} := T_1^{\alpha_1}T_2^{\alpha_2}\ldots T_m^{\alpha_m} :=
  \overbrace{T_1,T_1,\ldots,T_1}^{\alpha_1},
  \overbrace{T_2,T_2,\ldots,T_2}^{\alpha_2}, \ldots,
  \overbrace{T_m,T_m,\ldots,T_m}^{\alpha_m},
  \]
  where $\sum_{i=1}^m \alpha_i = l$ and $\alpha_i > 0 \forall i \in
  [m]$, and
  \[
  \mathbf{S} := S_1^{\beta_1}S_2^{\beta_2}\ldots S_n^{\beta_n} :=
  \underbrace{S_1,S_1,\ldots,S_1}_{\beta_1},
  \underbrace{S_2,S_2,\ldots,S_2}_{\beta_2}, \ldots,
  \underbrace{S_n,S_n,\ldots,S_n}_{\beta_n},
  \]
  where $\sum_{i=1}^n \beta_i = l$ and $\beta_i > 0 \forall i \in [n]$.
  Suppose that $\mathbf{T}$ and $\mathbf{S}$ satisfy the constraints:
  $T_{i} \neq T_{i+1} \forall i \in [m-1]$; $S_{i} \neq S_{i+1} \forall
  i \in [m-1]$; $n\leq m$; if $n = m$, then $S_i \neq T_i$ for some
  $i$. Then there is at least one block $T_i^{\alpha_i}$ of $\mathbf{T}$
  over which $\mathbf{T}$ and $\mathbf{S}$ mismatch everywhere;
  hence there are at least $\min\{\alpha_i: i \in [m]\}$ mismatches
  between the two sequences.
\end{lemma}

\begin{proof}
  Suppose that the claim is false; so in each block $T_i^{\alpha_i}$ of
  $\mathbf{T}$, there is a matching symbol in $\mathbf{S}$. Therefore,
  $T_1, T_2, \ldots, T_m$ is a subsequence of $\mathbf{S}$ and $n \geq
  m$. This, together with $T_{i} \neq T_{i+1} \forall i \in [m-1]$ and
  $n \leq m$, implies that $n = m$ and $S_i = T_i \forall i \in [m]$,
  which contradicts the assumption that when $n = m$, there is some $i$
  for which $S_i \neq T_i$.
\end{proof}

\begin{thm}\label{thm-paths}
  Let $P$ and $Q$ be any two pedigrees with $e(P) = e(Q) = 2e$ arcs. Let
  $\mathbf{T} := (T_i, i = 1,2,\ldots, m)$ be any sequence of undirected
  $X$-forests in which consecutive $X$-forests are non-isomorphic. Then
  for all $\mu \in (0,1/|\Sigma|)$, there exists $p_0 := p_0(e,m,\mu)\in
  (0,1)$ such that for all $p \in (0,p_0)$, the following statement is
  true: if $(\Sigma^{XL}: P, RM(p,\mu)) = (\Sigma^{XL}: Q, RM(p,\mu)) \,
  \forall \, L \in \nn$, then $n(\mathbf{G} > \mathbf{T}:P) =
  n(\mathbf{G} > \mathbf{T}:Q)$.
\end{thm}

\begin{proof} 
  Let $\mca := \mca_{1}^{m}:\mca_{2}^{m} :\ldots : \mca_m^{m}$, where
  $\mca_{i} := \mca(T_i, r_0, L)$ and $r_0$ are as defined in
  Equations~(\ref{eq-eps0}) and~(\ref{eq-mcai}), respectively. We will
  choose $\epsilon_{max}$ (defined at the end of
  Section~\ref{sec-consistency}) and $L$ (that depends on
  $\epsilon_{max}$) suitably later.  The probability of $\mca$ on $P$
  and $Q$ is written by summing (\ref{eq-pa}) over all $A \in \mca$. But
  based on Theorem~\ref{thm-identifiability},
  Corollary~\ref{cor-consistency} and Lemma~\ref{lem-eps-l2}, we can
  make the following qualitative and somewhat informal statement: If $L$
  is large enough, then $\pr\{\mca \cond P, RM(p,\mu)\}$ will get
  significantly higher contribution from spanning forest sequences
  $\mathbf{G}:= (G_i^{l_i}, i = 1,2,\ldots, m)$ of length $m^2L$ such
  that $T_u(G_i) \cong T_i$ for all $i \in [m]$ and $l_i$ are all close
  to $L$, than from spanning forest sequences $\mathbf{G}:= (G_i, i =
  1,2,\ldots, m^2 L)$ for which there are many mismatches (in terms of
  isomorphism) between sequences $(T_i^{mL}, i = 1,2,\ldots, m)$ and
  $((T_u(G_i)), i = 1,2,\ldots, m^2L)$ or if they require more than
  $m-1$ recombinations.

  Let $\pr\{\mca \cond P, RM(p,\mu)\} = \sum_{k\in \nn} P_k(P)$, where
  $P_k(P)$ is the joint probability of $\mca$ and the event that there
  are exactly $k$ recombinations. Furthermore, we write $P_{(m-1)}(P) =
  P_{(m-1)a}(P) + P_{(m-1)b}(P)$, where $P_{(m-1)a}(P)$ is the
  contribution from spanning forest sequences $\mathbf{G}:= (G_i^{l_i},
  i = 1,2,\ldots, m)$ of length $m^2 L$ such that $T_u(G_i) \cong T_i$
  for all $i \in [m]$, and $P_{(m-1)b}(P)$ is the remaining contribution
  to $P_{m-1}$, i.e., from spanning forest sequences $\mathbf{G} :=
  (G_i^{l_i}, i = 1,2,\ldots, l+1) $ of length $m^2 L$ such that either
  $l < m-1$ (i.e., the $m-1$ recombinations occur at fewer than $m-1$
  sites), or $l=m-1$ and $T_u(G_i) \not \cong T_i$ for some $i \in [m]$.

  We will show that only $P_{(m-1)a}(P)$ makes a significant
  contribution to $\pr\{\mca \cond P, RM(p,\mu) \}$. In particular, we
  will show that the various contributions to $\pr\{\mca \cond P,
  RM(p,\mu)\}$ take the following form:

  \begin{eqnarray*}
    P_{(m-1)a}(P) &=& n(\mathbf{G} \geq \mathbf{T}:P) \Delta(\mathbf{T})
    \\[1mm] 
    \sum_{k < m-1}P_k(P) + P_{(m-1)b}(P) &<& \delta_1 \\[1mm]
    \sum_{k \geq m} P_k(P) &<& \delta_2,
  \end{eqnarray*}
  where $\Delta(\mathbf{T})$ does not depend on the pedigree (but
  depends on the undirected $X$-forest sequence $\mathbf{T}$). Moreover,
  $\Delta(\mathbf{T})$, $\delta_1$ and $\delta_2$ depend on $e, m, L, p$
  and $\epsilon_{max}$.  We will show that when $p$ and $\epsilon_{max}$
  are sufficiently small and $L$ is sufficiently large, $\delta_1$ and
  $\delta_2$ are very small compared to $\Delta(\mathbf{T})$. It will
  imply that for $\pr\{\mca \cond P, RM(p,\mu)\}$ and $\pr\{\mca \cond
  Q, RM(p,\mu)\}$ to be equal, $n(\mathbf{G} \geq \mathbf{T}:P)$ and
  $n(\mathbf{G} \geq \mathbf{T}:Q)$ must be equal. (Otherwise, there
  would be a difference between $\pr\{\mca \cond P, RM(p,\mu)\}$ and
  $\pr\{\mca \cond Q, RM(p,\mu)\}$ that is of the order of a multiple of
  $\Delta(\mathbf{T})$.)

  \

  \noindent {\em A lower bound on $P_{(m-1)a}(P)$.}

  \

  Since the consecutive $X$-forests in $\mathbf{T}$ are nonisomorphic
  and $T_u(G_i) \cong T_i$ for all $i \in [m]$, there is at least one
  recombination between consecutive spanning forests $G_i$. And since
  there are exactly $m-1$ recombinations, there must be exactly one
  recombination between consecutive spanning forests.

  Therefore,
  \begin{eqnarray}\label{eq-term1a}
    P_{(m-1)a}(P) &=& n(\mathbf{G} \geq \mathbf{T}:P)
    \left(\frac{(1-p)^{e(m^2L-m)+(m-1)(e-1)}p^{m-1}}{2^{e}}\right) 
    \nonumber \\[1mm]
    && \times \left( \sum_{A \in \mca} \sum_{\substack{\mathbf{l} 
          \in \zz_+^m : \\
          \sum\limits_{i=1}^ml_i = m^2L}}\;
      \prod_{i=1}^m\pr\{A[L_{i-1}+1,L_i] \cond T_u(G_i),M(\mu)
      \}\right) \nonumber \\
    &:=& n(\mathbf{G} \geq \mathbf{T}:P) \Delta(\mathbf{T}),
  \end{eqnarray}
  where $L_0 := 0$, $L_i := L_{i-1} + l_i$. The inner summation is
  evaluated for any fixed choice of $\mathbf{G} := (G_i^{l_i}, i =
  1,2,\ldots, m)$ such that $T_u(G_i) \cong T_i$ for all $i \in [m]$,
  because for any fixed $l_i, i \in [m]$, and spanning forest sequences
  $\mathbf{G} := ((G_i)^{l_i}, i = 1,2,\ldots, m)$ and
  $\mathbf{G^{\prime}} := ((G_i^{\prime})^{l_i}, i = 1,2,\ldots, m)$, if
  $T_u(G_i) \cong T_u(G_i^{\prime})$ for all $i \in [m]$, then for each
  alignment $A$, we have $\pr\{A[L_{i-1}+1,L_i] \cond G_i^{l_i},M(\mu)
  \} = \pr\{A[L_{i-1}+1,L_i] \cond (G_i^{\prime})^{l_i},M(\mu) \} =
  \pr\{A[L_{i-1}+1,L_i] \cond T_u(G_i),M(\mu) \}$. Therefore, the RHS is
  a product of $n(\mathbf{G} \geq \mathbf{T}:P)$ and a factor that does
  not explicitly depend on the pedigree, but only on $\mathbf{T}$.

  We have 
  \begin{eqnarray}\label{eq-term1b}
    \Delta(\mathbf{T}) 
    &\geq &
    \frac{(2cL)^{m-1}(1-p)^{em^2 L-m-e+1}p^{m-1} 
      (1-\epsilon_{max})^{(m^2)}}{2^{e}}.
  \end{eqnarray}
  To prove the lower bound, we sum (\ref{eq-pa}) over spanning forest
  sequences of the type $\mathbf{G}:= (G_i^{l_i}, i = 1,2,\ldots, m)$
  where $T_u(G_i) \cong T_i$ for all $i \in [m]$ and $l_i$ are such that
  the $m-1$ recombinations occur at sites $L_i \in [imL-cL,imL+cL), i =
  1\;\text{to}\;m-1$, for a small positive constant $c$, (i.e., we
  ignore contributions from the spanning forest sequences in which some
  of the recombinations are not near the boundaries of the blocks
  $\mca_{i}^m$ of $\mca$). We can choose the recombination sites in
  $(2cL)^{m-1}$ ways. Since $s(\mathbf{G}) = em^2L-e-m+1$ and
  $r(\mathbf{G}) = m-1$, the probability of $\mathbf{G}$ is
  $(1-p)^{em^2L-m-e+1}p^{m-1}/2^e$. We write $\mathbf{G} :=
  (\mathbf{G_i}, i = 1,2,\ldots, m)$, where $\mathbf{G_i}$ are blocks of
  $\mathbf{G}$ of length $mL$ each. Then we have
  \[
  \pr\{\mca_{i}^m \cond \mathbf{G_i}, M(\mu) \} \geq
  (1-\epsilon_{max})^m
  \]
  for each $\mathbf{G}$ and for all $i$, provided we choose $c, L$ and
  $\epsilon_{max}$ appropriately according to Lemma~\ref{lem-eps-l2}.
  
  \
  
  \noindent {\em An upper bound on $\sum_{k < m-1}P_k(P) +
    P_{(m-1)b}(P)$.}

  \

  We have
  \begin{eqnarray}\label{eq-3b}
    && \sum_{k < m-1}P_k(P) + P_{(m-1)b}(P)  \nonumber \\[2mm]
    & \leq & \sum_{k=0}^{m-1} \frac{1}{2^e}
    \left(\sum_{l=0}^{k}n_{kl}\binom{m^2 L -1}{l}
      (\epsilon_{max})^{m-l}\right) p^k(1-p)^{e(m^2L-1)-k} \nonumber \\[2mm]
    & \leq & c_2 L^{m-1} \epsilon_{max} := \delta_1,
  \end{eqnarray}
  where $n_{kl}$ is the number of spanning forest sequences $(G_1,G_2,
  \ldots, G_{l+1})$ in which $k$ recombinations occur at $l$ sites, and
  $c_2 > 0$ is a constant that depends on $e$ and $m$. We explain below
  how the bound is obtained.

  We fix $\mathbf{G} \in \mcg_P^{m^2 L}$ such that $r(\mathbf{G}) = k
  \leq m-1$. For $k=m-1$, since we are only interested in the
  contribution $P_{(m-1)b}$, we fix $\mathbf{G}$ with the additional
  restrictions in the definition of $P_{(m-1)b}$. Suppose that the $k$
  recombinations occur at $l$ distinct sites $L_1, L_2, \ldots, L_{l}$,
  where $0 = L_0 < L_1 < L_2 < \ldots < L_l < L_{l+1} = m^2 L$. So we
  write $\mathbf{G} := (G_1^{L_1},G_2^{(L_2-L_1)},\ldots ,G_{l+1}^{(m^2
    L - L_l)})$. There are $n_{kl}$ choices for $(G_1,G_2, \ldots,
  G_{l+1})$. For each choice of $(G_1,G_2, \ldots, G_{l+1})$, there are
  $\binom{m^2 L -1}{l}$ choices for the recombining sites $L_1, \ldots ,
  L_l$. Each $\mathbf{G}$ has a probability
  $(1-p)^{s(\mathbf{G})}p^{r(\mathbf{G})}/2^e$, where $r(\mathbf{G}) =
  k$ and $s(\mathbf{G}) = e(m^2 L - 1) - k$. We show that $\pr\{\mca
  \cond \mathbf{G}, M(\mu) \}$ is bounded above by
  $(\epsilon_{max})^{m-l}$ for each $\mathbf{G}$ with $l$ recombining
  sites that satisfies the above constraints.

  Let $\mathbf{G} := (G_1, G_2, \ldots, G_{m^2L})$ be a spanning forest
  sequence of length $m^2L$. For $i \in [m], j \in [m]$, we refer to the
  subsequences $\mathbf{G_i} := (G_k, k \in [(i-1)mL+1, imL])$ as {\em
    blocks} of $\mathbf{G}$, and subsequences $\mathbf{G_{ij}} := (G_k,
  k \in [((i-1)m+j-1)L+1, ((i-1)m+j)L])$ as {\em subblocks} of
  $\mathbf{G}$. We say that subblock $\mathbf{G_{ij}}$ is {\em
    recombination-free} if there are no recombinations between any two
  sites of the subblock. (The subblock may have recombinations at its
  boundaries.)

  By Lemma~\ref{lem-seq}, there is at least one block, say the $i$-th
  block, over which the sequences $\mathbf{T}$ and
  $(T_u(G_1)^{L_1},T_u(G_2)^{(L_2-L_1)},\ldots ,T_u(G_{l+1})^{(m^2 L -
    L_l)})$ mismatch everywhere. Since there are $m$ subblocks in each
  block and $l \leq m-1$, there are at least $m-l$ recombination-free
  subblocks in the $i$-th block of $\mathbf{G}$. Let these subblocks be
  denoted by $\mathbf{G_{ij_k}} := (G_{ij_k})^L, k = 1,2, \ldots$. We
  have $T_i \not \cong T_u(G_{ij_k})$ for each of them. Therefore,
  \[
  \pr\{\mca \cond \mathbf{G}, M(\mu) \} \leq \prod_{k} \pr\{\mca_i \cond
  T_u(G_{ij_k}, M(\mu) \} \leq (\epsilon_{max})^{m-l}.
  \]

  \
  
  \noindent {\em An upper bound on $\sum_{k \geq m} P_k(P)$.}
  
  \
  
  We use the following fact about binomially distributed random
  variables: If $X \sim \text{Bin}(n,p)$, then $\pr\{X \geq k \} \leq
  \binom{n}{k}p^k$. (This result is a consequence of the union bound.)
  Since there are $e(m^2 L-1)$ points at which a recombination can
  possibly occur, we have
  \begin{equation}
    \sum_{k \geq m} P_k(P) \leq \binom{e(m^2 L-1)}{m}p^{m} 
    \leq c_3 L^m p^m := \delta_2,
  \end{equation}
  where $c_3 > 0$ is a constant that depends on $e$ and $m$.

  We write similar bounds for $Q$.

  Now suppose that $\pr\{\mca \cond P,RM(p,\mu) \} = \pr\{\mca
  \cond Q,RM(p,\mu)\}$ but $n(\mathbf{G} \geq \mathbf{T}:P) \neq
  n(\mathbf{G} \geq \mathbf{T}:Q)$. Therefore, Equations~\ref{eq-term1a}
  and \ref{eq-term1b} imply that
  \begin{eqnarray}\label{eq-term1c}
    && |P_{(m-1)a}(P) - P_{(m-1)a}(Q)| \nonumber \\[2mm]
    &=& |n(\mathbf{G} \geq \mathbf{T}:P) - n(\mathbf{G} \geq
    \mathbf{T}:Q)|\Delta(\mathbf{T}) \nonumber \\[2mm]
    &\geq& \frac{(2cL)^{m-1}(1-p)^{(em^2 L-m-e+1)}p^{(m-1)}
      (1-\epsilon_{max})^{(m^2)}}{2^{e}}.
  \end{eqnarray}
  But that is impossible if we can choose $\epsilon_{max}$, $L$, 
  and $p$ so that $\delta_1 + \delta_2 < \Delta(\mathbf{T})$, or
  \begin{equation}\label{eq-final}
    c_2L^{m-1}\epsilon_{max} + c_3 (Lp)^m
    < \frac{(2cLp)^{m-1}(1-p)^{(em^2 L-m-e+1)}
      (1-\epsilon_{max})^{(m^2)}} {2^{e}}.
  \end{equation}
  In other words, the discrepancy $|P_{(m-1)a}(P) - P_{(m-1)a}(Q)|$
  cannot be compensated for by the remaining terms in $\pr\{\mca \cond
  P,RM(p,\mu) \}$ and $\pr\{\mca \cond Q,RM(p,\mu) \}$. Such a choice is
  possible. For example, we first choose $\epsilon_{max} \in (0,1)$, let
  us say, $\epsilon_{max} = 1/M$, where $M > 1$. We then set $L := L(M)$
  and $p := p(M)$ so that
  \begin{eqnarray}\label{eq-cond}
    (\epsilon_{max}/p^{m-1}) \longrightarrow 0 
    &\;\; \text{as}\;\; & M \longrightarrow \infty , \nonumber \\
    Lp  \longrightarrow 0 &\;\; \text{as}\;\; & M
    \longrightarrow \infty, \;\;\text{and}\nonumber \\
    (1-p)^ {(em^2 L-m-e+1)}(1-\epsilon_{max})^{(m^2)}
    \longrightarrow 1 &\;\; \text{as}\;\; & M
    \longrightarrow \infty.
  \end{eqnarray}

  The choice of $L$ must guarantee phylogenetic consistency as in
  Corollary~\ref{cor-consistency}. Moreover, we must set $L$ as small as
  possible so as to get a better bound on $p$. By Lemma~\ref{lem-eps-l2},
  $L$ is required to grow only logarithmically in $1/\epsilon_{max}$, so
  let $L := c\log M$ for a suitable choice of $c > 0$. Now conditions
  (\ref{eq-cond}) are satisfied for $p := (\log M)^{-(1+\epsilon)}$ for
  $\epsilon > 0$. Then (\ref{eq-final}) is satisfied for sufficiently
  large $M$.
\end{proof}

\subsection{Applications}
In the following, we illustrate Theorem~\ref{thm-paths} with a general
result on counting $X$-forests and a couple of specific examples.
 
\begin{cor} If the conditions of Theorem~\ref{thm-paths} are
  satisfied, then $P$ and $Q$ have the same number of undirected
  $X$-forests of each type.
\end{cor}

\begin{proof}
  We apply Theorem~\ref{thm-paths} for $m = 1$. Suppose that $P$ is a
  pedigree with $e(P) = 2e$ arcs. Let $T$ be an undirected
  $X$-forest. Suppose that we want to count the number of copies $n(T:P)$
  of $T$ in P. Let $T_i, i = 1 \; \text{to} \; n(T:P)$ be the distinct
  copies of $T$ in $P$. For each $T_i$, let $T_{ij}, j = 1,2,\ldots $ be
  the undirected $X$-forests in $P$ that contain $T_i$. (To be precise, we
  have subgraphs $S_i$ and $S_{ij}$ in $P$ such that when the (directed)
  arcs of $S_i$ and $S_{ij}$ are replaced by (undirected) edges, we get
  the undirected $X$-forests $T_i$ and $T_{ij}$. But we do not make this
  distinction in the following.)  Note that $T_{ij}$ are all distinct
  $X$-forests. There are $e-e(T_i)$ non-founder vertices in $P$ at which we
  can choose one of the two incoming arcs to construct a spanning forest
  $G$ that contains $T_i$. Therefore, for each $T_i$ there are
  $2^{e-e(T_i)} = 2^{e-e(T)}$ spanning forests that contain
  $T_i$. Therefore,
  
  \begin{equation*}
    n(T:P) 2^{e-e(t)} = \sum_i \sum_j |\{G \in \mcg_P : T_u(G) 
    = T_{ij}\}|.
  \end{equation*}
  Now by grouping terms on the RHS by isomorphism classes of $T_{ij}$,
  we obtain
  \begin{equation*}
    n(T:P) 2^{e-e(T)} = \sum_{(T^\prime \in \lVert \mcu_P \rVert)\,
      \wedge \, (T^\prime \geq T)}n(G >T^\prime:P).
  \end{equation*}
  Since $n(G >T^\prime:P) = n(G >T^\prime:Q)$ for all undirected $X$-forests
  $T^\prime$, we also have $n(T:P) = n(T:Q)$ for all undirected
  $X$-forests $T$.
\end{proof}

\begin{cor} The examples of non-reconstructible pedigrees given in
  \cite{thatte18} can be distinguished from the probability
  distribution on the space of alignments under Model RM.
\end{cor}

\begin{proof} Pedigrees in these examples do not have the same number of
  $X$-forests of each type. For example, one pedigree in each pair contains
  a common ancestor of all extant vertices while the other does
  not. This may also be observed in the examples in Figure~\ref{fig-k3}.
\end{proof}

But knowing the number of $X$-forests of each type is not always enough
to distinguish pedigrees, as verified by pedigrees $P_1$ and $Q_1$ shown
in Figure~\ref{fig-hap}. They both have the same number of directed
(therefore, also undirected) $X$-forests of each type: they have 4
directed $X$-forests in which 1 and 2 have a common ancestor. We denote
their undirected $X$-forest by $T_1$, which is a path of length 4 with
end vertices 1 and 2. There are 12 directed $X$-forests in which 1 and 2
do not have a common ancestor. We denote their undirected $X$-forest by
$T_2$, which consists of two isolated vertices 1 and 2. Moreover, $n(G >
T_1:P_1) = n(G > T_1:Q_1) = 16$ and $n(G > T_2:P_1) = n(G > T_2:Q_1) =
48$. Also, for $\mathbf{T} := (T_1, T_2)$ (and for $\mathbf{T} := (T_2,
T_1)$), we check by direct counting that $n(\mathbf{G} > \mathbf{T}:P_1)
= n(\mathbf{G} > \mathbf{T}:Q_1) = 64$. But $P_1$ and $Q_1$ can
nevertheless be distinguished as shown below.

\begin{cor}\label{cor-pq}
  Pedigrees $P_1$ and $Q_1$ in Figure~\ref{fig-hap} are distinguished
  from the probability distribution on the space of alignments under
  Model RM provided the crossover probability is sufficiently small.
\end{cor}

\begin{proof} Let $T_1$ and $T_2$ be the $X$-forests as described
  above. We apply Theorem~\ref{thm-paths} for $m = 3$ with
  $\mathbf{T}:=(T_1,T_2,T_1)$. We can verify that $n(\mathbf{G} >
  \mathbf{T}:P_1) = 112$ and $n(\mathbf{G} > \mathbf{T}:Q_1) =
  104$. Therefore, $P_1$ and $Q_1$ give different distributions on the
  space of alignments.  In the following, we describe how $n(\mathbf{G}
  > \mathbf{T}:P_1)$ and $n(\mathbf{G} > \mathbf{T}:Q_1)$ are counted.

  \ 

  Let $T_a$ denote the directed $X$-forests in $P_1$ and $Q_1$
  consisting of paths $a\cdots 1$ and $a\cdots 2$. Similarly, we write
  $T_b,T_c,T_d$ for other directed $X$-forests in $P_1$ and $Q_1$,
  rooted at $b$, $c$ and $d$, respectively. These are the 4 directed
  $X$-forests that have $T_1$ as their undirected $X$-forest. All other
  directed $X$-forests have $T_2$ as their undirected $X$-forest.

  \
  
  \noindent {\em Counting} $n(\mathbf{G} > \mathbf{T}:P_1)$: Since
  $T_u(G_1) \cong T_u(G_3) \cong T_1$, we count 16 different
  contributions to $n(\mathbf{G} > \mathbf{T}:P_1)$ depending on the
  choices for $T_d(G_1)$ and $T_d(G_3)$ in $\{T_a, T_b, T_c,
  T_d\}$. Please refer to Figure~\ref{fig-p}.

  \begin{figure}[ht]
    \begin{center}
      \includegraphics[width=118mm]{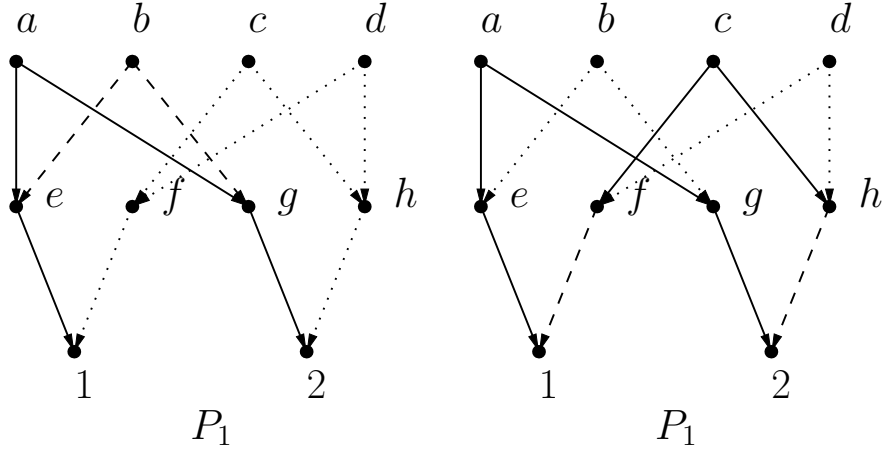}
    \end{center}
    \caption[]{Counting $n(\mathbf{G} > \mathbf{T}:P_1)$ when $T_d(G_1) = T_a$
      and $T_d(G_3) = T_b$ (left), and when $T_d(G_1) = T_a$ and
      $T_d(G_3) = T_c$ (right)}
    \label{fig-p}
  \end{figure}

  When $T_d(G_1) = T_a$ and $T_d(G_3) = T_a$: This is possible only if
  $G_3 = G_1$. There are 4 choices of $G_1$, and for each of them there
  are 4 choices for $G_2$ depending on which arc of $T_a$ is replaced to
  obtain $G_2$. Thus we have 16 sequences $\mathbf{G}$ for which $T_d(G_1)
  = T_a$ and $T_d(G_3) = T_a$.

  When $T_d(G_1) = T_a$ and $T_d(G_3) = T_b$: the dashed arcs $bg$ and
  $be$ in Figure~\ref{fig-p} on the left must be obtained by replacing
  $ag$ and $ae$, and there are two sequences $\mathbf{G}$ that achieve
  this: either $G_2 = G_1 - ae + be$ or $G_2 = G_1 - ag + bg$. Also,
  there are 4 possible ways to include arcs pointing to $f$ and $h$ in
  $G_1$. Therefore, there are 8 sequences $\mathbf{G}$ such that
  $T_d(G_1) = T_a$ and $T_d(G_3) = T_b$.

  When $T_d(G_1) = T_a$ and $T_d(G_3) = T_c$: there are only 2 sequences
  $\mathbf{G}$ for which this is possible, since the arcs $cf$ and $cg$
  (shown in bold in Figure~\ref{fig-p} on the right) must already be in
  $G_1$.

  When $T_d(G_1) = T_a$ and $T_d(G_3) = T_d$: this case is similar to the
  case $T_d(G_3) = T_c$.

  Thus there are 28 choices of $\mathbf{G}$ such that $T_d(G_1) = T_a$,
  and similarly 28 choices each for $T_d(G_1) = T_b$, $T_d(G_1) = T_c$,
  and $T_d(G_1) = T_d$. Therefore, there are 112 sequences $\mathbf{G}$
  such that $\mathbf{G} > \mathbf{T}$.

  \

  \noindent {\em Counting} $n(\mathbf{G} > \mathbf{T}:Q_1)$: Again we
  count 16 different contributions to $n(\mathbf{G} > \mathbf{T}:Q_1)$
  depending on the choices for $T_d(G_1)$ and $T_d(G_3)$ in $\{T_a, T_b,
  T_c, T_d\}$.  Please refer to Figure~\ref{fig-q}.

  \begin{figure}[ht]
    \begin{center}
      \includegraphics[width=118mm]{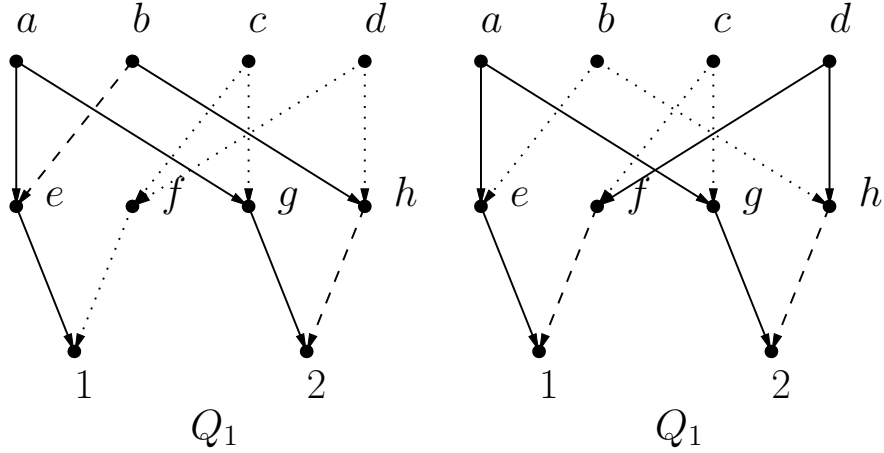}
    \end{center}
    \caption[]{Counting $n(\mathbf{G} > \mathbf{T}:Q_1)$ when $T_d(G_1)
      = T_a$ and $T_d(G_3) = T_b$ (left), and when $T_d(G_1) = T_a$ and
      $T_d(G_3) = T_d$ (right)}
    \label{fig-q}
  \end{figure}

  When $T_d(G_1) = T_a$ and $T_d(G_3) = T_a$: as in case of $P_1$, we
  have 16 sequences $\mathbf{G}$ for which $T_d(G_1) = T_a$ and
  $T_d(G_3) = T_a$.

  When $T_d(G_1) = T_a$ and $T_d(G_3) = T_b$: the dashed arcs $be$ and
  $h2$ in Figure~\ref{fig-q} on the left must be obtained by replacing
  $ae$ and $g2$, and arc $bh$ must already be in $G_1$. There are two
  sequences $\mathbf{G}$ that achieve this: either $G_2 = G_1 - ae + be$
  or $G_2 = G_1 - g2 + h2$. Also, there are 2 choices for arcs pointing
  to $f$, therefore, there are 2 choices for $G_1$. Therefore, there are
  4 sequences $\mathbf{G}$ such that $T_d(G_1) = T_a$ and $T_d(G_3) =
  T_b$.

  When $T_d(G_1) = T_a$ and $T_d(G_3) = T_c$: this case is similar to
  the case in which $T_d(G_3) = T_b$, therefore, there are 4 sequences
  such that $T_d(G_1) = T_a$ and $T_d(G_3) = T_c$.

  When $T_d(G_1) = T_a$ and $T_d(G_3) = T_d$: in this case, the arcs
  $df$ and $dh$ must already be in $G_1$. Thus there are only two
  sequences $\mathbf{G}$ that are counted depending on whether $G_2 =
  G_1 - e1 + f1$ or $G_2 = G_1 - g2 + h2$.

  Thus there are 26 choices of $\mathbf{G}$ such that $T_d(G_1) = T_a$,
  and similarly 26 choices each for $T_d(G_1) = T_b$, $T_d(G_1) = T_c$,
  and $T_d(G_1) = T_d$. Therefore, there are 104 sequences $\mathbf{G}$
  such that $\mathbf{G} > \mathbf{T}$.

  Thus we have verified that $n(\mathbf{G} > \mathbf{T}:P_1)$ and
  $n(\mathbf{G} > \mathbf{T}:Q_1)$ are unequal for $\mathbf{T}
  :=(T_1,T_2,T_1)$, implying that $P_1$ and $Q_1$ are distinguished by
  the probability distribution they induce on the extant sequences under
  model RM.
\end{proof}

\section{Discussion and open questions}\label{sec-problems}

In this paper we have presented a rigorous mathematical framework for
studying pedigree reconstruction problems under probabilistic models. We
extended phylogenetic identifiability results to reconstruct pedigrees
under an idealised model of recombination and mutation. The main result
of this paper is the computation of a class of combinatorial invariants
from the joint distribution of extant sequences. As a corollary, we were
able to show that certain known examples of pedigrees that could not be
distinguished from path lengths alone, could be distinguished by a more
detailed analysis of subgraph sequences. Here we identify a few open
problems and directions for future investigation.

\noindent{\em Identifiability of computationally tractable invariants:}
The invariants of a pedigree $P$ defined by $n(\mathbf{G} >
\mathbf{T}:P)$ may be quite difficult to apply in
general. Theorem~\ref{thm-paths} may be difficult to use computationally
for bigger pedigrees. Even for other reconstruction problems in graph
theory, computational verifications are difficult. For example, Ulam's
reconstruction conjecture has been computationally verified to be true
only for graphs on at most 11 vertices \cite{mckay97}. Even on
restricted classes of graphs, computational reconstruction experiments
are difficult to perform. It will be useful to derive from $n(\mathbf{G}
> \mathbf{T}:P)$ (or independently) other identifiable invariants that
may be easier to use in computational experiments.

Theorem~\ref{thm-paths} only states that if two pedigrees induce the
same joint distribution on extant sequences under model RM, then they
agree on the invariants $n(\mathbf{G} > \mathbf{T})$. But it will be
important to prove a converse or a result of the type: two pedigrees
induce the same distribution on extant sequences under model RM if and
only if they take the same value for a class of combinatorial
invariants. Such a result would reduce the identifiability problem to a
purely combinatorial problem of proving or disproving that the class of
invariants is complete. Such a class of invariants may be $n(\mathbf{G}
> \mathbf{T})$ or it may be somewhat stronger than $n(\mathbf{G} >
\mathbf{T})$. It may be possible to compute, by the methods of
Theorem~\ref{thm-paths}, other invariants, for example, spanning forest
sequences in which consecutive spanning forests are not necessarily
separated by just one recombination.

\noindent {\em Improving the bound on $p$:}
In the lower bound on $\Delta(\mathbf{T})$ in
Equation~(\ref{eq-term1b}), we have $1/2^e$ in the denominator, while
$c_3$ is roughly $e^m$, where $e$ is the number of arcs in a pedigree.
Therefore, it is possible to obtain better bounds on $p$ for the
applicability of the main theorem for pedigrees with fewer arcs.
Therefore, we would significantly improve the upper bound on $p$ if we
showed that a pedigree can be reconstructed from the collection of its
subpedigrees (pedigrees of subsets of the extant population) of order
$k$ for some small $k$. In \cite{thatte18} we made a conjecture about
how small $k$ may be for pedigrees of order $n$ in which the population
remains constant over generations. Thus solving the purely combinatorial
problem of reconstructing pedigrees from their subpedigrees, while
important in its own right, will be useful for improving bounds on $p$
in Theorem~\ref{thm-paths}.

On the other hand it is also likely (although we do not conjecture) that
Theorem~\ref{thm-paths} is valid without restrictions on $p$, or for all
but finitely many values of $p$, or for all $p$ except when $|\Sigma|$
takes small values. But we do not have good intuition as to why
Proposition~\ref{prop-n2-2} (with $|\Sigma| = 2$) requires a more
complicated argument and an upper bound on $p$.

{\em Maximum likelihood computation of the invariants $n(\mathbf{G} >
  \mathbf{T})$}: It will be of interest to derive statistical
consistency results and bounds on sequence lengths, analogous to
Corollary~\ref{cor-consistency} and Lemma~\ref{lem-eps-l1}, for
computing the invariants $n(\mathbf{G} > \mathbf{T})$. For example, we
would like to make the following qualitative statement precise. Suppose
$P$ is a pedigree with $e$ arcs. Let $\epsilon > 0$ be given. Suppose
$\mathbf{T}$ is an undirected $X$-forest sequence of length $m$ with no
two consecutive $X$-forests isomorphic. Then there is a sufficiently
large $L_{\epsilon, m}$ such that if a collection of sequences of length
$L > L_{\epsilon,m}$ evolved on $P$ giving an alignment $A$, then the
likelihood ratio $L(A \cond P)/L(A \cond Q)$ is large for all pedigrees
$Q$ such that $n(\mathbf{G} > \mathbf{T}:P) \neq n(\mathbf{G} >
\mathbf{T}:Q)$.  We expect that $L_{\epsilon, m}$ would be of the order
of $m \log (1/\epsilon)$.

In the model RM, we assumed that the founders are independently assigned
sequences from a uniform distribution. This assumption may be relaxed or
replaced by more realistic assumptions.

\subsection*{Acknowledgements}
  I would like to thank Jotun Hein and Mike Steel for many valuable
  comments and for pointing me to many useful references relevant to
  this work. In particular, Mike's comments on and references to the
  literature on phylogenetic identifiability results and results on the
  consistency of ML were very useful. This project was funded by the
  grant ``From complete genomes to global pedigrees'' from EPSRC
  (Engineering and Physical Sciences Research Council), U.K.  I would
  like to thank Jotun Hein and EPSRC for supporting my postdoctoral stay
  at the University of Oxford under this grant. I am currently working
  on the project ``Tree-graphs and incidence matrices: theory and
  applications'' funded by CNPq, Brasil (Processo: 151782/2010-5) at the
  Instituto de Matem\'atica e Estat\'istica, Universidade de S\~{a}o
  Paulo, where I completed a significant revision of this paper and
  improvements on the bounds in the main theorem. I would also like to
  thank Yoshiharu Kohayakawa with whom I had very useful conversations
  during the revision of this paper.

\bibliographystyle{plain}

\begin{thenomenclature} 

 \nomgroup{A}

  \item [{$(\Sigma^{XL}:P,RM(p,\mu))$}]\begingroup the space of alignments on   the extant vertices of $P$ as a probability space under model   $RM(p,\mu)$; other probability spaces are denoted analogously\nomeqref {3.0}
		\nompageref{14}
  \item [{$[a,b], [a,b), \ldots$}]\begingroup intervals in integers   and reals\nomeqref {2.0}
		\nompageref{6}
  \item [{$[m]$}]\begingroup the set $\{1,2,\ldots, m\}$\nomeqref {2.0}
		\nompageref{6}
  \item [{$\cong$}]\begingroup isomorphism between pedigrees, $X$-forests,   graphs, etc.\nomeqref {2.0}
		\nompageref{7}
  \item [{$\lVert \mct_P \rVert, \lVert \mcu_P \rVert$}]\begingroup the sets of   isomorphism classes of directed and undirected $X$-forests (or   distinct directed or undirected $X$-forests) in a pedigree $P$,   respectively\nomeqref {2.0}
		\nompageref{10}
  \item [{$\lVert S \rVert$}]\begingroup - isomorphism class of an object; if $S$   is a class of objects, then it is the set of isomorphism classes of   objects in the class\nomeqref {2.0}
		\nompageref{7}
  \item [{$\lVert T \rVert$}]\begingroup the isomorphism class of an $X$-forest   $T$\nomeqref {2.0}
		\nompageref{10}
  \item [{$\mathbf{f}(A) := (f_1, f_2, \ldots, f_N)$}]\begingroup vector of   fractional site pattern frequencies in an alignment $A$\nomeqref {5.0}
		\nompageref{23}
  \item [{$\mathbf{G}:=(G_1,G_2,\ldots, G_m)$}]\begingroup - a spanning forest   sequence of length $m$\nomeqref {2.0}
		\nompageref{10}
  \item [{$\mathbf{p}(T, \mu):= (p_1,p_2, \ldots,p_N)$}]\begingroup defined as   $p_i := \pr\{C_i \cond T,M(\mu) \}$\nomeqref {5.0}
		\nompageref{23}
  \item [{$\mathbf{T}:=(T_1,T_2,\ldots, T_m)$}]\begingroup - an $X$-forest   sequence of length $m$\nomeqref {2.0}
		\nompageref{10}
  \item [{$\mca :=   \mca_{1}:\mca_{2}:\ldots:\mca_m$}]\begingroup a set of alignments obtained by   concatenating alignments from sets $\mca_{i},  i = 1,2,\ldots,m $\nomeqref {5.4}
		\nompageref{29}
  \item [{$\mca(T_i,r_0,L)$}]\begingroup the set of alignments of length $L$   whose fractional site pattern frequencies are within a radius $r_0$   from $\mathbf{p}(T_i)$\nomeqref {5.2}
		\nompageref{24}
  \item [{$\mca, \mca_{i}$}]\begingroup subsets of an alignment space such as   $\Sigma^{XL}$\nomeqref {3.0}
		\nompageref{15}
  \item [{$\mcg_P$}]\begingroup the set of spanning forests in a pedigree $P$\nomeqref {2.0}
		\nompageref{10}
  \item [{$\mct_P, \mcu_P$}]\begingroup the sets of directed and undirected   $X$-forests in a pedigree $P$, respectively\nomeqref {2.0}
		\nompageref{10}
  \item [{$\mu$}]\begingroup the substitution probability in   models $RM(p,\mu)$ and $M(\mu)$\nomeqref {3.0}
		\nompageref{14}
  \item [{$\nn$}]\begingroup the set of natural numbers\nomeqref {2.0}
		\nompageref{6}
  \item [{$\pr\{.\}$}]\begingroup the probability of an event\nomeqref {3.0}
		\nompageref{15}
  \item [{$\rho(\mathbf{s}, r)$}]\begingroup a ball of   radius $r$ centred at $\mathbf{s}$\nomeqref {5.0}
		\nompageref{23}
  \item [{$\Sigma$}]\begingroup finite alphabet\nomeqref {2.0}
		\nompageref{8}
  \item [{$\Sigma^X$}]\begingroup the set of site patterns on $X$\nomeqref {2.0}
		\nompageref{8}
  \item [{$\Sigma^{XL}$}]\begingroup the set of alignments of length $L$ on $X$\nomeqref {2.0}
		\nompageref{8}
  \item [{$\zz$}]\begingroup the set of integers\nomeqref {2.0}
		\nompageref{6}
  \item [{$\zz_+$}]\begingroup the set   of positive integers\nomeqref {2.0}
		\nompageref{6}
  \item [{$A,A_i,\ldots$ }]\begingroup alignments on $X$, i.e., maps from $X$ to   $\Sigma^L$ or elements of $\Sigma^{XL}$ \nomeqref {2.0}
		\nompageref{8}
  \item [{$A:=A_1:A_2:\ldots:A_m$}]\begingroup an alignment obtained by   concatenating alignments $A_i, i = 1,2,\ldots,m$\nomeqref {5.4}
		\nompageref{29}
  \item [{$C$, $C_i$, ...}]\begingroup characters on $X$, i.e., maps $C:X   \rightarrow \Sigma$\nomeqref {2.0}
		\nompageref{8}
  \item [{$d(\mathbf{x},\mathbf{y})$}]\begingroup 1-norm distance between   $\mathbf{x}$ and $\mathbf{y}$\nomeqref {5.0}
		\nompageref{23}
  \item [{$d^-(u), d^+(u), d(u)$}]\begingroup - in-degree, out-degree, degree (or   total degree) of a vertex $u$\nomeqref {2.0}
		\nompageref{7}
  \item [{$G \cong H$}]\begingroup - $G$ and $H$ are isomorphic\nomeqref {2.0}
		\nompageref{7}
  \item [{$G \leq H$ or $H \geq G$}]\begingroup - when used for graphs (or   isomorphism classes of graphs) $G$ and $H$, it means $G$ is isomorphic   to a subgraph of $H$\nomeqref {2.0}
		\nompageref{7}
  \item [{$G \subseteq H$ or $H \supseteq G$}]\begingroup - when used for labelled   graphs $G$ and $H$, it means $G$ is a subgraph of $H$\nomeqref {2.0}
		\nompageref{7}
  \item [{$G,G_i, \ldots$}]\begingroup - spanning forests in a pedigree\nomeqref {2.0}
		\nompageref{10}
  \item [{$n(\mathbf{G} > \mathbf{T}:P)$}]\begingroup number of sequences   $\mathbf{G}$ of spanning forests in $P$ for which $T_u(G_i) \cong T_i$   for all $G_i$ in $\mathbf{G}$ and consecutive $G_i$ are separated by   exactly 1 recombination\nomeqref {5.4}
		\nompageref{29}
  \item [{$p$}]\begingroup the crossover probability in models $R(p)$ and   $RM(p,\mu)$\nomeqref {3.0}
		\nompageref{14}
  \item [{$P$, $Q$, $P(X,Y,U,E)$, ...}]\begingroup pedigrees\nomeqref {2.0}
		\nompageref{7}
  \item [{$r(\mathbf{G})$}]\begingroup the number of recombinations in   $\mathbf{G}$, see Definition~\ref{def-rg-sg}\nomeqref {2.0}
		\nompageref{11}
  \item [{$s(\mathbf{G})$}]\begingroup the number of points of no recombination in   $\mathbf{G}$, see Definition~\ref{def-rg-sg}\nomeqref {2.0}
		\nompageref{11}
  \item [{$S^k $}]\begingroup the set of $k$-tuples of elements of a set $S$\nomeqref {2.0}
		\nompageref{6}
  \item [{$S^X$}]\begingroup the set of all functions from $X$ to $S$\nomeqref {2.0}
		\nompageref{6}
  \item [{$T,T_i, \ldots$}]\begingroup - $X$-forests in a pedigree or $X$-forests\nomeqref {2.0}
		\nompageref{10}
  \item [{$T_d(G)$}]\begingroup - the unique directed $X$-forest in a spanning   forest $G$ in a pedigree\nomeqref {2.0}
		\nompageref{10}
  \item [{$T_u(G)$}]\begingroup - the unique undirected $X$-forest in a spanning   forest $G$ in a pedigree\nomeqref {2.0}
		\nompageref{10}
  \item [{$u \leq v$ - }]\begingroup (for vertices $u$ and $v$ in a pedigree)   there is a directed path from $v$ to $u$\nomeqref {2.0}
		\nompageref{7}
  \item [{$V(G), E(G)$}]\begingroup - vertex and edge sets of a graph,   respectively\nomeqref {2.0}
		\nompageref{7}
  \item [{$v(G), e(G)$}]\begingroup - cardinalities of vertex and edge sets of a   graph, respectively\nomeqref {2.0}
		\nompageref{7}
  \item [{$X$}]\begingroup the set of extant vertices of a pedigree\nomeqref {2.0}
		\nompageref{7}

\end{thenomenclature}

\end{document}